\documentclass[english]{article}
\usepackage[T1]{fontenc}
\usepackage[utf8]{inputenc}
\usepackage[a4paper]{geometry}
\usepackage{float}
\usepackage{amsmath}
\usepackage{amsthm}
\usepackage{amssymb}
\usepackage{nameref}

\makeatletter

\providecommand{\tabularnewline}{\\}

\newcommand{\lyxaddress}[1]{
	\par {\raggedright #1
	\vspace{1.4em}
	\noindent\par}
}
\theoremstyle{plain}
\newtheorem{thm}{\protect\theoremname}
\theoremstyle{remark}
\newtheorem*{acknowledgement*}{\protect\acknowledgementname}

\usepackage{graphicx}
\usepackage{palatino}
\usepackage[T1]{fontenc}

\usepackage{thmtools}

\usepackage{cite}

\usepackage{hyperref}
\hypersetup{colorlinks=true,
            linkcolor=blue,
            filecolor=magenta,      
            urlcolor=cyan
           }

\makeatother

\usepackage{babel}
\providecommand{\acknowledgementname}{Acknowledgement}
\providecommand{\theoremname}{Theorem}

\begin{document}
\global\long\def\ket#1{|#1\rangle}%
 
\global\long\def\bra#1{\langle#1|}%
 
\global\long\def\braket#1#2{\langle#1|#2\rangle}%
 
\global\long\def\bk#1#2{\langle#1|#2\rangle}%
 
\global\long\def\kb#1#2{|#1\rangle\langle#2|}%

\global\long\def\c#1{\mathbb{C}^{#1}}%

\global\long\def\abs#1{\mid#1\mid}%

\global\long\def\avg#1{\langle#1\rangle}%
\global\long\def\norm#1{\Vert#1\Vert}%
\global\long\def\normcb#1{\Vert#1\Vert_{0}}%

\title{\vspace{-1cm}
Uncertainty relations\\
for the support of quantum states}
\author{Vincenzo Fiorentino and Stefan Weigert}
\date{January 2023}
\maketitle

\lyxaddress{\begin{center}
Department of Mathematics, University of York\\
York YO10 5DD, United Kingdom\\
vincenzo.fiorentino@york.ac.uk$\qquad$stefan.weigert@york.ac.uk
\par\end{center}}
\begin{abstract}
Given a narrow signal over the real line, there is a limit to the
localisation of its Fourier transform. In spaces of prime dimensions,
Tao derived a sharp state-independent uncertainty relation which holds
for the support sizes of a pure qudit state in two bases related by
a discrete Fourier transform. We generalise Tao's uncertainty relation
to complete sets of mutually unbiased bases in spaces of prime dimensions.
The bound we obtain appears to be sharp for dimension three only.
Analytic and numerical results for prime dimensions up to nineteen
suggest that the bound cannot be saturated in general. For prime dimensions
two to seven we construct sharp bounds on the support sizes in $(d+1)$
mutually unbiased bases and identify some of the states achieving
them.
\end{abstract}
\tableofcontents{}

\section{\label{sec:Introduction}Introduction}

No quantum particle can reside in a state with both its position and
momentum distributions being localised arbitrarily well. For these
incompatible observables, Heisenberg's uncertainty relation \cite{heisenberg_uber_1927,kennard_zur_1927}
establishes a finite lower bound for the product of their variances.
This result relies on a fundamental property of Fourier theory: a
real (or complex) function with \emph{finite} support on the real
line has a Fourier transform which must be non-zero almost everywhere
\cite{folland_uncertainty_1997}. It is, however, difficult to quantify
the support of functions on unbounded intervals. Using variances instead
of the supports of probability distributions circumvents this difficulty.

The situation is different for quantum systems with finite-dimensional
Hilbert spaces since the \emph{support (size)} of a pure state---defined
as the number of non-zero components in a given orthonormal basis---is
always finite. A computational basis state in $\c d$, say, has support
equal to one, and the support of its (discrete) Fourier transform
equals $d$ since all basis states contribute. Thus, the \emph{product}
of the support sizes equals $d$ which turns out to be its smallest
possible value \cite{donoho_uncertainty_1989}.

The underlying product inequality has been generalised in a number
of directions \cite{meshulam_uncertainty_1992,wigderson_uncertainty_2021}.
Tao derived an \emph{additive} inequality \cite{tao_uncertainty_2005}
which is valid in spaces $\c d$ of prime dimensions $d=p$: the \emph{sum}
of the supports of a state and its discrete Fourier transform is bounded
from below by the value $(p+1)$. This bound is sharp since a computational
basis state and its Fourier transform saturate it.

Support inequalities and their generalisations have found applications
in signal processing \cite{candes_robust_2006,candes_quantitative_2006},
for example, and they can be used to identify non-classical quantum
states \cite{de_bievre_complete_2021} whose \textit{\emph{Kirkwood-Dirac
quasiprobability distribution}} \cite{kirkwood_quantum_1933,dirac_analogy_1945}
is \textit{not} a probability distribution. Such states provide an
advantage in quantum metrology \cite{arvidsson-shukur_quantum_2020}
and play a role in weak measurements \cite{hofmann_role_2011,dressel_significance_2012,arvidsson-shukur_conditions_2021}
and contextuality \cite{kunjwal_anomalous_2019}.

Using the support size of a quantum state as a measure of uncertainty
has an unexpected---and previously unnoticed---operational advantage.
Quantum supports take a \emph{finite} set of integer values only,
in stark contrast to other measures. Variances of observables in a
given state or its von Neumann entropy take \emph{real numbers} as
values which demands many measurements to determine experimentally.
However, a \emph{finite} number of measurements may already suffice
to determine the \emph{exact} support size of a quantum state. This
situation occurs whenever the state at hand has $(i)$ ``full''
support in the basis considered and $(ii)$ each outcome has been
registered at least once. Conformity with a given support inequality
may be verified by a \emph{finite }number of measurements as small
as the bound itself. This property depends, of course, on the assumption
that outcomes with probability zero never occur; limited detection
efficiency does not invalidate the argument, however.

Variance-based uncertainty relations also exist for more than two
observables associated with multiple orthonormal bases \cite{kechrimparis_heisenberg_2014,kechrimparis_geometry_2017,dodonov_generalized_1980,dodonov_variance_2018}:
position and momentum may be supplemented by a third continuous variable
which is canonical to each of them. The eigenbases of these three
observables are mutually unbiased and related by fractional Fourier
transforms. The product of their variances satisfies a \emph{triple
uncertainty relation} \cite{kechrimparis_heisenberg_2014}. Importantly,
the lower bound of this inequality does \emph{not} follow from the
pair uncertainty relations but must be determined independently. No
quantum state exists which satisfies all three pair uncertainty relations
simultaneously. In a similar vein, entropic uncertainty relations
capture the incompatibility of up to $(d+1)$ observables in finite-dimensional
systems, linked to a complete set of mutually unbiased bases known
to exist in prime-power dimension \cite{ivanovic_inequality_1992,sanchez_entropic_1993,ballester_entropic_2007}. 

The main goal of this paper is to extend Tao's additive support uncertainty
relation to the case of more than two bases, inspired by the triple
uncertainty relation for continuous variables. The focus will be on
prime-dimensional spaces where $(p+1)$ mutually unbiased bases exist,
known as \emph{complete sets}. The support sizes of a state in any
pair of mutually unbiased bases from such a set are expected to satisfy
Tao's bound but it is unlikely that they will saturate all pair bounds
simultaneously.

This paper is structured as follows. Sec. \ref{sec: Preliminaries}
sets up notation by briefly describing known product and sum inequalities
for the support of a vector, and the properties of complete sets of
MU bases are summarised. In Sec. \ref{sec:The-extended-support-inequality},
Tao's additive uncertainty relation for the support of quantum states
is shown to hold for any pair of mutually unbiased bases in the complete
set considered, and the generalised additive support inequality involving
all $(p+1)$ mutually unbiased bases is established as a direct consequence.
According to Sec. \ref{sec:Saturating-the-inequality} the bounds
provided by the generalised support inequality cannot be saturated
for prime dimensions $2\leq d\leq19$, except $d=3$. Higher \emph{achievable}
bounds are derived in Sec. \ref{sec:Achievable-bounds}, for prime
numbers up to $d=7$. In the last section, we summarise and discuss
the results obtained. The proofs of some lemmata are relegated to
an appendix.

\section{Preliminaries \label{sec: Preliminaries}}

\subsection{Support inequalities for a Fourier pair of bases \label{subsec: Support-inequality-Fourier-pairs}}

The \emph{support }(\emph{size}) of a Hilbert-space vector $\psi\in{\cal H}_{d}$
is given by the number of its non-zero expansion coefficients $\psi_{v}=\bk v{\psi}$
in an orthonormal basis ${\cal B}=\left\{ \ket v,v=0,1,\ldots,\,d-1\right\} $,
\begin{equation}
|\text{supp}(\psi,\mathcal{B})|=\#(\psi_{v}\neq0,v=0\ldots d-1)\in\left\{ 0\ldots\,d\right\} \,.\label{eq: define support}
\end{equation}
The only vector with vanishing support is the zero vector. Due to
normalisation, the support of a quantum state must be at least one,
and the maximum is achieved whenever the state $\psi$ is a linear
combination of all $d$ basis states. The support size of a state
clearly depends on the chosen basis. Formally, the support size can
be obtained as a limit of the Rényi entropy \cite{renyi_measures_1961}
of the probability distribution $\left\{ |\psi_{v}|^{2},v=0\ldots d-1\right\} $.

Thinking of the support size as the (improper) $L^{0}$-``norm''
of $\psi$, we will use the notation 
\begin{equation}
|\text{supp}(\psi,\mathcal{B})|=\norm{\psi}_{\mathcal{B}}\,.\label{eq: define zero-norm}
\end{equation}

The set of expansion coefficients $\left\{ \psi_{v},v=0\ldots d-1\right\} $
has three obvious \emph{support-conserving} symmetries. The support
size is invariant $(i)$ under rephasing each expansion coefficient
separately,

\begin{equation}
||\psi||_{\mathcal{B}}=||R\psi||_{\mathcal{B}}\,,\qquad R=\text{diag}(e^{i\tau_{0}},e^{i\tau_{1}},\ldots,e^{i\tau_{d-1}})\,,\label{eq: phase symmetry}
\end{equation}
with real numbers $\tau_{v},v=0,\ldots,\,d-1$; $(ii)$ under permuting
the components of any state among themselves 
\begin{equation}
||\psi||_{\mathcal{B}}=||P\psi||_{\mathcal{B}}\,,\qquad P\in S_{d}\,,\label{eq: perm symmetry}
\end{equation}
where $S_{d}$ is the permutation group acting on sets of $d$ elements;
and $(iii)$ under the complex conjugation of some (or all) of its
components,
\begin{equation}
||\psi||_{\mathcal{B}}=||K\psi||_{\mathcal{B}}\,,\qquad K=\prod_{\text{some }v\in\left\{ 0\ldots\,d-1\right\} }K_{v}\,,\label{eq: conjug symmetry}
\end{equation}
where each operator $K_{v},v=0,\ldots,\,d-1$, maps one expansion
coefficient of the state $\psi$ in the basis $\mathcal{B}$ to its
complex conjugate, $K_{v}\psi_{v}=\psi_{v}^{*}$, and does not change
the others. In the basis $\mathcal{B}$, the permutations $P$ are
represented by a matrix of order $d$ containing exactly one unit
entry in each row and column; hence the unitary invariances of rephasing
and permuting coefficients are conveniently combined into \emph{monomial}
matrices $M\equiv RP$. The third invariance described by the operators
$K_{v}$ will play no role.

Given two distinct orthonormal bases ${\cal B}$ and ${\cal B}^{\prime}$
of ${\cal H}_{d}$, one may ask to which extent a state can be ``localised''
in both of them. Clearly, the product of its support sizes in ${\cal B}$
and ${\cal B}^{\prime}$ may take values between one and $d^{2}$.
If the bases are related by $\mathcal{B}^{\prime}=F\mathcal{B}$,
where $F$ is the discrete Fourier transform with matrix elements
(in the $\mathcal{B}$-basis)
\begin{equation}
F_{vv^{\prime}}=\frac{1}{\sqrt{d}}e^{-2\pi ivv^{\prime}/d}\qquad v,v^{\prime}\in\left\{ 0\ldots\,d-1\right\} \,,\label{eq: Fourier matrix}
\end{equation}
then the \textit{product} of the support sizes of a state $\psi$
and its Fourier transform $\widetilde{\psi}=F^{\dagger}\psi$ is bounded
from below \cite{donoho_uncertainty_1989},
\begin{equation}
\norm{\psi}_{\mathcal{B}}\,||\psi||_{\mathcal{B}^{\prime}}\geq d\,,\label{eq:Product UR (Donoho=000026Stark) - Introduction}
\end{equation}
where we use the fact that the support size of the Fourier transformed
state $\widetilde{\psi}$ in the basis $\mathcal{B}$ is equal to
the support size of the state $\psi$ in the basis $\mathcal{B}^{\prime}$,
i.e. 
\begin{equation}
||\widetilde{\psi}||_{\mathcal{B}}=||F^{\dagger}\psi||_{\mathcal{B}}=||\psi||_{\mathcal{B}^{\prime}}\,.\label{eq: Fourier support notation}
\end{equation}
The inequality (\ref{eq:Product UR (Donoho=000026Stark) - Introduction})
represents a finite-dimensional equivalent of Heisenberg's uncertainty
relation for position and momentum observables of a quantum particle:
quantum states localised in position, say, necessarily come with a
broad variance in momentum, the Fourier-transformed position observable.

For spaces ${\cal H}_{d}$ with \emph{prime }dimensions $d$, an \emph{additive}
inequality for the supports of a quantum state in a pair of Fourier-related
bases is known \cite{tao_uncertainty_2005},
\begin{equation}
\norm{\psi}_{\mathcal{B}}+||\psi||_{\mathcal{B}^{\prime}}\geq d+1\,,\label{eq:Sum UR (Tao) - Introduction}
\end{equation}
which is stronger than the multiplicative relation (\ref{eq:Product UR (Donoho=000026Stark) - Introduction}),
as follows from the inequality $d+1-x\geq d/x$, for $x\in[1,d]$.
\textit{\emph{In the terminology of }}\cite{de_bievre_complete_2021},
any two bases ${\cal B}$ and ${\cal B}^{\prime}$ are said to be
\emph{completely incompatible} if and only if the support sizes of
the expansion coefficients of any (non-zero) vector $\psi\in\mathcal{H}_{d}$
satisfy this bound.

The inequality (\ref{eq:Sum UR (Tao) - Introduction}) is a special
case of a theorem valid for finite additive Abelian groups $G$ with
$|G|$ elements and trivial subgroups only which necessitates the
restriction to prime dimensions \cite{tao_uncertainty_2005}. Consider
a complex-valued function $f:G\to\mathbb{C}$ and its transform $\widetilde{f}:G\to\mathbb{C}$,
defined by 
\begin{equation}
\widetilde{f}(v^{\prime})=\frac{1}{\sqrt{|G|}}\sum_{v\in G}f(v)\overline{e(v,v^{\prime})}\,,\label{eq:Fourier transform of f-2-1}
\end{equation}
where $e(v,v^{\prime})$ is a ``bi-character'' of $G$ satisfying
$e(v_{1}+v_{2},v^{\prime})=e(v_{1},v^{\prime})e(v_{2},v^{\prime})$
and an analogous relation for its second argument. In the particular
case of $e(v,v^{\prime})=e^{-2\pi ivv^{\prime}/d}$, one obtains an
inequality for the supports of $f$ and $\widetilde{f}\equiv F^{\dagger}f$.
\begin{thm}[\emph{Tao's theorem }]
\label{thm:Tao's-theorem-2-1}If $f:G\to\mathbb{C}$ is a non-zero
function and the cardinality $|G|$ of the group $G$ is prime, then
\begin{equation}
|\text{\emph{supp}}(f)|+|\text{\emph{supp}}(\widetilde{f})|\geq|G|+1\,.\label{eq:additive UR simple}
\end{equation}
\end{thm}
Upon identifying $f\left(v\right)$ with \textbf{$\braket v{\psi}$}
and $\widetilde{f}(v^{\prime})$ with $\braket{v^{\prime}}{\psi}$,
respectively, we obtain the inequality (\ref{eq:Sum UR (Tao) - Introduction})
relative to the bases ${\cal B}$ and $\mathcal{B}^{\prime}$ introduced
via $F$ in Eq. (\ref{eq: Fourier matrix}).

The main ingredient of Tao's proof is a fundamental property of the
Fourier matrix in prime dimensions \cite{stevenhagen_chebotarev_1996,frenkel_simple_2004,tao_uncertainty_2005}
which dates back to the 1920s: all its square submatrices are invertible.
\begin{thm}[\emph{Chebotarëv's theorem}]
\label{thm:Chebotarev's theorem} If $d$ is prime, then all minors
of the Fourier matrix $F$ in Eq. (\ref{eq: Fourier matrix}) are
non-zero.
\end{thm}
The inequalities (\ref{eq:Product UR (Donoho=000026Stark) - Introduction})
and (\ref{eq:Sum UR (Tao) - Introduction}) involve a \emph{pair}
of \textit{\emph{mutually unbiased}}\emph{ }bases of $\mathcal{H}_{d}$,
namely the computational basis ${\cal B}$ and its Fourier transform.
We will now introduce larger sets of mutually unbiased bases to formulate
more general support inequalities. Not surprisingly, Chebotarëv's
theorem must be generalised to other matrices which emerge when establishing
bounds on support sizes of quantum states in multiple bases (cf. Sec.
\ref{subsec:Support-inequality-MU-pairs}). 

\subsection{Mutually unbiased bases in prime dimensions \label{sec:Mutually-unbiased-bases}}

Two orthonormal bases of the space $\mathcal{H}_{d}=\mathbb{C}^{d}$
are said to be \textit{mutually unbiased}\textit{\emph{ (MU)}} if
the inner products between any two states (not of the same basis)
have modulus $1/\sqrt{d}$. Then, to know the outcome of a projective
measurement performed in one basis implies complete uncertainty about
the outcome of a subsequent projective measurement performed in the
other.

When $d$ is a power of a prime number $p$, i.e. $d=p^{n}$, sets
of $(d+1)$ mutually unbiased bases have been constructed \cite{ivanovic_geometrical_1981,wootters_optimal_1989,bandyopadhyay_new_2002,durt_about_2005}.
Such complete sets are both \textit{maximal}---in the sense that
no additional MU basis can be added to it---and \textit{tomographically
complete}\textit{\emph{:}}\emph{ }the probability distributions of
outcomes in the $(d+1)$ bases uniquely encode an unknown quantum
state. It is an open problem whether complete sets of MU bases exist
in composite dimensions, $d\neq p^{n}$.

For $d=2$, the eigenstates of the Pauli operators $Z_{2}$, $X_{2}$
and $X_{2}Z_{2}=-iY_{2}$ form a complete set which has a simple structure.
Representing the computational basis ${\cal B}_{0}$ by the identity
matrix $H_{0}=I_{\left(2\times2\right)}$, the following two Hadamard
matrices encode the bases which are MU to ${\cal B}_{0}$,
\begin{equation}
H_{1}=F=\frac{1}{\sqrt{2}}\begin{bmatrix}1 & 1\\
1 & -1
\end{bmatrix},\qquad\qquad H_{2}=DF\quad\text{where}\;D=\begin{bmatrix}1 & 0\\
0 & i
\end{bmatrix}\,.\label{eq:Hadamard matrices in d=00003D2}
\end{equation}

If $d$ is an odd prime, the eigenstates of the $(d+1)$ generalised
Pauli operators $Z_{d}$, $X_{d}$, $X_{d}Z_{d}$, $X_{d}Z_{d}^{2}$,
..., $X_{d}Z_{d}^{d-1}$, represent a maximal set of MU bases. The
$k$-th state of the $j$-th basis is given by 
\begin{equation}
\ket{\phi_{k}^{j}}=\frac{1}{\sqrt{d}}\sum_{x=0}^{d-1}\omega^{-kx}\omega^{\left(j-1\right)x^{2}}\ket{\phi_{x}^{0}}\,,\qquad j\in\left\{ 1...\,d\right\} \,,\quad k\in\left\{ 0...\,d-1\right\} \,,\label{eq:MU basis states odd prime d}
\end{equation}
where the states $\left\{ \ket{\phi_{x}^{0}}\equiv\ket x,x=0\ldots d-1\right\} $
form the computational basis ${\cal B}_{0}$ and $\omega\equiv e^{2i\pi/d}$
is a $d$-th root of the number $1$ \cite{durt_about_2005}. For
each value of $j$, the equimodular expansion coefficients 
\begin{equation}
\left[H_{j}\right]_{xk}=\langle x\ket{\phi_{k}^{j}}=\frac{1}{\sqrt{d}}\omega^{-kx}\omega^{\left(j-1\right)x^{2}}\,,\qquad j\in\left\{ 1...\,d\right\} ,\,x,k\in\left\{ 0...\,d-1\right\} \,,\label{eq:hadamard matrices defn}
\end{equation}
define a\emph{ complex-valued Hadamard matrix}. These are unitary
matrices since their columns are given by the components (in the computational
basis ${\cal B}_{0}$) of $d$ orthogonal vectors. When combined with
the computational basis ${\cal B}_{0}$, the states given in Eq. (\ref{eq:MU basis states odd prime d})
form a complete set of MU bases for Hilbert spaces of prime dimensions,
which we will refer to as the \emph{standard }set. In this paper,
all MU bases will be taken from the standard set.

The Hadamard matrix $H_{1}$ in (\ref{eq:hadamard matrices defn})
coincides with the discrete Fourier matrix $F$ given in Eq. (\ref{eq: Fourier matrix}).
The remaining Hadamard matrices $H_{j}$ map the computational basis
${\cal B}_{0}$ of ${\cal H}_{d}$ to other orthonormal bases denoted
by ${\cal B}_{j}$. Adopting an active view of these transformations,
the state $\psi$ is mapped to the state $H_{j}^{\dagger}\psi$. The
relation between the supports of the state $\psi$ in $\mathcal{B}_{0}$
and the $j$-th MU basis $\mathcal{B}_{j}$ reads,
\begin{equation}
||H_{j}^{\dagger}\psi||_{0}=||\psi||_{j}\,,\label{eq: support notation}
\end{equation}
abbreviating the notation introduced in (\ref{eq: Fourier support notation}),
i.e. $||\psi||_{\mathcal{B}_{j}}\equiv||\psi||_{j}$, $j\in\{0\ldots d\}$.

The columns of the $d$ Hadamard matrices $H_{j}$ in (\ref{eq:hadamard matrices defn})
are related in a simple way to each other, namely by 
\begin{equation}
\ket{\phi_{k}^{j}}=D^{j-1}B^{k}\ket{\phi_{0}^{1}}\,,\qquad j\in\left\{ 1...\,d\right\} ,\,k\in\left\{ 0...\,d-1\right\} ,\label{eq: relating columns of H_j}
\end{equation}
with two diagonal $(d\times d)$ matrices $B$ and $D$; in other
words, all states of the complete set of MU bases can be generated
easily from any given state such as $\ket{\phi_{0}^{1}}$---except
for the states of the computational basis ${\cal B}_{0}$. Within
each Hadamard matrix, the matrix $B$ cyclically shifts a given column
to the right, 
\begin{equation}
B\ket{\phi_{k}^{j}}=\begin{cases}
\ket{\phi_{k+1}^{j}}\,, & k=0,\ldots,\,d-2\,,\\
\ket{\phi_{0}^{j}}\,, & k=d-1\,;
\end{cases}\label{eq:B matrix acting on states}
\end{equation}
its entries are given by the components of the second column of the
Fourier matrix $F=H_{1}$, 
\begin{equation}
B=\text{diag}\left(1,\omega^{-1},\ldots,\omega^{-(d-1)}\right)\,,\label{eq: def of B}
\end{equation}
except for the factor $\sqrt{d}$.

The matrix $D$ is given by the components of the first column of
the second Hadamard matrix $H_{2}$, i.e.
\begin{equation}
D=\text{diag}\left(1,\omega^{1},\ldots,\omega^{(d-1)^{2}}\right)\,,\label{eq: def of D}
\end{equation}
cyclically mapping a state of the $j$-th MU basis to the corresponding
one of the MU basis with label $(j+1)$,
\begin{equation}
D\ket{\phi_{k}^{j}}=\begin{cases}
\ket{\phi_{k}^{j+1}}\,, & j=1,\ldots,\,d-1\,,\\
\ket{\phi_{k}^{1}}\,, & j=d\,.
\end{cases}\label{eq:D matrix acting on states}
\end{equation}
In terms of Hadamard matrices, this property reads
\begin{equation}
DH_{j}=\begin{cases}
H_{j+1}\,, & j=1,\ldots,\,d-1\,,\\
H_{1}\,, & j=d\,.
\end{cases}\label{eq:D matrix acting on Hadamard}
\end{equation}

Writing $H_{j}=D^{j-1}H_{1}\equiv D^{j-1}F$, Chebotarëv's\emph{ }theorem
is seen to imply that the minors of all Hadamard matrices $H_{j}$,
$j=1\ldots d$, are non-zero: the ranks of the minors of $F$ do not
change upon multiplying their rows with non-zero scalars. In view
of Eq. (\ref{eq:Hadamard matrices in d=00003D2}), this generalisation
is also valid for the case of $d=2$.

\section{Support inequalities$\ldots$ \label{sec:The-extended-support-inequality}}

How well can one localise quantum states simultaneously in $(d+1)$
MU bases? To answer this question, we need to minimise the support
sizes of a state relative to the complete set. In a first step, we
now show that the support sizes of a quantum state relative to any
\emph{pair} of standard MU bases also satisfy Tao's bound (\ref{eq:Sum UR (Tao) - Introduction}).
Second, by combining the resulting pair inequalities, we establish
a state-independent lower bound.

\subsection{$\ldots$ for arbitrary pairs of MU bases\label{subsec:Support-inequality-MU-pairs}}

Tao's result establishes---for a space of prime dimension $d$ supporting
a cyclic abelian group---a sharp inequality for the support sizes
of a quantum state and its Fourier transform. Most pairs of the MU
bases introduced in Eq. (\ref{eq:MU basis states odd prime d}) are,
however, not related by a Fourier transform. Nevertheless, Tao's bound
also holds for the supports of the images of any quantum state generated
by two Hadamard matrices as we will show now.
\begin{thm}
\label{thm: pair sums}Given any pair of distinct standard MU bases
associated with matrices $H_{j}$ and $H_{k}$, $j,k\in\left\{ 0\ldots\,d\right\} $,
the support sizes of a quantum state $\ket{\psi}\in{\cal H}_{d}$
satisfy the state-independent sharp bound, 
\begin{equation}
\norm{\psi}_{j}+\norm{\psi}_{k}\geq d+1\,,\label{eq: all pair sum URs}
\end{equation}
where $d$ is any prime number.
\end{thm}
It is important to realise that Theorem \ref{thm: pair sums} does
not cover \emph{arbitrary} pairs of MU bases in prime dimensions but
only those defined in Eq. (\ref{eq:hadamard matrices defn}). Nevertheless,
\textit{all }\textit{\emph{pairs}} of MU bases in dimensions $d=2,3$
and $5$ are found to be completely incompatible since in these dimensions
\emph{all }Hadamard matrices are equivalent to the Fourier matrix.
Already for the next prime, $d=7$, other types of Hadamard matrices
exist \cite{bruzda_catalogue_2006}.
\begin{proof}
The case of dimension $d=2$ is straightforward. If a state $\ket{\psi}\in{\cal H}_{2}$
has support one in one MU basis, it must have support two in both
other bases, due to being MU to their members. Thus, the sum of the
supports of any state in two bases must be at least three.

For odd primes $d,$we will consider two cases separately: either
\emph{$(i)$} one of the bases in Eq. (\ref{eq: all pair sum URs})
is the computational basis, so that $j=0$, say, or \emph{$(ii)$}
neither of them. 

$(i)$ Defining the vector $\phi=D^{1-k}\psi$, we obtain
\begin{equation}
\norm{\psi}_{0}+\norm{\psi}_{k}=\norm{\psi}_{0}+\norm{H_{k}^{\dagger}\psi}_{0}=\norm{\psi}_{0}+\norm{F^{\dagger}D^{1-k}\psi}_{0}=\norm{D^{k-1}\phi}_{0}+\normcb{F^{\dagger}\phi}\,,\label{eq:new Lemma 4 - part 1}
\end{equation}
recalling that $H_{k}=D^{k-1}F$ holds according to Eq. (\ref{eq:D matrix acting on Hadamard}).
Since $D$ is a diagonal unitary hence a support-conserving unitary
matrix (cf. (\ref{eq: phase symmetry})), we obtain
\begin{equation}
\norm{\psi}_{0}+\norm{\psi}_{k}=\norm{\phi}_{0}+\normcb{F^{\dagger}\phi}\geq d+1\,,\label{eq:new Lemma 4 - part 2}
\end{equation}
where Tao's theorem was used in the last step.

$(ii)$ Now consider the case where the non-zero labels $j$ and $k$
differ from each other. Defining the vector $\phi=H_{j}^{\dagger}\psi$,
the sum of the support sizes can be written as
\begin{equation}
\norm{\psi}_{j}+\norm{\psi}_{k}=\normcb{H_{j}^{\dagger}\psi}+\norm{H_{k}^{\dagger}\psi}_{0}=\normcb{\phi}+\normcb{H_{k}^{\dagger}H_{j}\phi}\,.\label{eq:new Lemma 4 - part 3}
\end{equation}
The product $H_{k}^{\dagger}H_{j}$ of two distinct Hadamard matrices
is, in fact, always equal to another Hadamard matrix $H_{t}^{\dagger}$,
$t\neq0$, up to a monomial matrix $M(j,k)$; this result is the content
of Lemma \ref{lem: reduction of product of Hadamards} stated directly
after the proof. As the matrix $M(j,k)$ is support-conserving (cf.
Eqs. (\ref{eq: phase symmetry}) and (\ref{eq: perm symmetry})) for
all states of the space $\mathcal{H}_{d}$, we find 
\begin{equation}
\normcb{\phi}+\normcb{M(j,k)H_{t}^{\dagger}\phi}=\normcb{\phi}+\normcb{H_{t}^{\dagger}\phi}\geq d+1\,,\label{eq:new Lemma 4 - part 4}
\end{equation}
where (\ref{eq:new Lemma 4 - part 2}) of Part $(i)$ has been used
in the final step.
\end{proof}
The proof just relies on dissolving products of the Hadamard matrices
$H_{j}$ which encode a complete set of MU bases. Clearly, products
of the form $H_{k}^{\dagger}H_{j}$, $j\neq k$, are Hadamard matrices
since their matrix elements, being overlaps of MU vectors, have modulus
$1/\sqrt{d}$,
\begin{equation}
\left[H_{k}^{\dagger}H_{j}\right]_{\ell\ell^{\prime}}=\braket{\phi_{\ell}^{k}}{\phi_{\ell^{\prime}}^{j}}\,.\label{eq:product of Hadamard matrices -- elements}
\end{equation}
When $d=2$, one finds explicitly that $H_{1}^{\dagger}H_{2}=MH_{2}^{\dagger}$
and $H_{2}^{\dagger}H_{1}=M^{\prime}H_{2}^{\dagger}$, with monomial
matrices $M$ and $M^{\prime}$. In other words, the phases of the
matrix elements (\ref{eq:product of Hadamard matrices -- elements})
coincide with those of the adjoint of another transition matrix after
permuting and rephasing its rows. This property actually holds for
any prime dimension.

\begin{restatable}{lem}{productHadamardlemma}

\label{lem: reduction of product of Hadamards}Let $d$ be an odd
prime and $j,k\in\left\{ 1\ldots\,d\right\} $ with $j\neq k$. Then
\begin{equation}
H_{k}^{\dagger}H_{j}=M(j,k)H_{t}^{\dagger}\label{eq:new lemma (j neq k, both nonzero)}
\end{equation}
for a monomial matrix $M(j,k)$ if and only if $t=1+\chi\in\left\{ 1,\ldots,d\right\} $
where the integer $\chi$ satisfies $4\left(j-k\right)\chi=1\mod d$.

\end{restatable}
\begin{proof}
See \nameref{Appendix}.
\end{proof}
Furthermore, Lemma \ref{lem: reduction of product of Hadamards} allows
us to generalise Chebotarëv's theorem to the product matrices $H_{k}^{\dagger}H_{j}$,
for distinct indices $j$ and $k$.

\begin{restatable}[]{cor}{CheboThmHadamard}

\label{cor:Extension of Chebo to product matrices}If $d$ is prime,
then all minors of the Hadamard matrices $H_{k}^{\dagger}H_{j}$,
$j,k\in\{0\ldots d\}$ and $j\neq k$, are non-zero.

\end{restatable}
\begin{proof}
Let $d$ be an odd prime. Any $H_{t}$, $t\in\{1\ldots d\}$, has
only non-zero minors, as was mentioned after Eq. (\ref{eq:D matrix acting on Hadamard}),
as do the adjoints $H_{t}^{\dagger}$. Therefore, the claim holds
if one of the labels $j,k$, is zero. For both $j,k\neq0$, Lemma
\ref{lem: reduction of product of Hadamards} applies. Since rephasing
and permuting the rows of a matrix do not change the rank of any submatrix,
we can conclude that the matrices $H_{k}^{\dagger}H_{j}$, $j,k\neq0$
and $j\neq k$ also have non-zero minors only.

In dimension $d=2$, the result follows from inspecting the products
$H_{1}^{\dagger}H_{2}=MH_{2}^{\dagger}$ and $H_{2}^{\dagger}H_{1}=M^{\prime}H_{2}^{\dagger}$.
\end{proof}
According to Corollary \ref{cor:Extension of Chebo to product matrices}
the vectors formed by the columns (or rows) of all square submatrices
of the Hadamard matrices $H_{k}^{\dagger}H_{j}$, $j\neq k$, are
linearly independent. What is more, up to $d$ vectors taken from
\textit{any} \textit{two} MU bases are linearly independent.

\begin{restatable}[]{cor}{linearindepcorollary}

\label{cor:Linear independence of d vectors from ANY two MUBs}Given
a complete standard set of MU bases in the space $\mathcal{H}_{d}$
of prime dimension $d$, up to $d$ vectors taken from any two MU
bases are linearly independent.

\end{restatable}
\begin{proof}
See \nameref{Appendix}.
\end{proof}
Theorem \ref{thm: pair sums} also demonstrates that all pairs of
MU bases taken from the complete standard set in prime dimension are
\textit{completely}\textit{\emph{ }}\textit{incompatible}\textit{\emph{,
in the sense of Sec.}}\textit{ }\ref{subsec: Support-inequality-Fourier-pairs}\textit{.}
This statement is stronger than the results of \cite{tao_uncertainty_2005,de_bievre_complete_2021}
since the bases we consider are not necessarily related by the Fourier
matrix $F$.

\subsection{$\ldots$ for complete sets of $(d+1)$ MU bases}

Let us denote the sum of the numbers of non-zero expansion coefficients
of a state $\psi\in\mathcal{H}_{d}$ in a complete standard set of
MU bases by

\begin{equation}
{\cal S}(d)=||\psi||_{0}+||\psi||_{1}+\dots+||\psi||_{d}\,.\label{eq:define overall support S(d)}
\end{equation}
Then, the inequalities (\ref{eq: all pair sum URs}) imply that the
overall support size $\mathcal{S}(d)$ cannot fall below a certain
threshold. This is a central result of our paper.
\begin{thm}
\label{thm:multi-bound}For any prime $d$, the overall support ${\cal S}(d)$
of a quantum state $\ket{\psi}\in{\cal H}_{d}$ in a complete standard
set of MU bases satisfies the additive state-independent bound
\begin{equation}
{\cal S}(d)\geq\frac{\left(d+1\right)^{2}}{2}\equiv T(d)\,.\label{eq:extended sum UR}
\end{equation}
\end{thm}
\begin{proof}
Write down $(d+1)$ copies of the support inequality (\ref{eq: all pair sum URs})
with indices $(j,j+1)$, $j=0,\ldots\,d-1$, and $(d,0)$, respectively.
Adding them up, the right-hand-sides of (\ref{eq: all pair sum URs})
give $(d+1)^{2}$, and since each term $||\psi||_{j}\equiv\norm{H_{j}^{\dagger}\psi}_{0}$,
$j=0,\ldots,\,d$, occurs twice on the left, one obtains the inequality
(\ref{eq:extended sum UR}).
\end{proof}
An alternative proof treats all pair supports equally: write down
(\ref{eq: all pair sum URs}) for all $d(d+1)$ distinct pairs of
indices $(j,k)$ and consider the sum of the supports. After removing
common factors, the bound (\ref{eq:extended sum UR}) on ${\cal S}(d)$
follows.

Support inequalities other than Eqs. (\ref{eq: all pair sum URs})
and (\ref{eq:extended sum UR}) exist. They may involve any number
between two and $(d+1)$ MU bases. For example, picking the first
three MU bases and combining the associated pair inequalities from
(\ref{eq: all pair sum URs}) leads to the additive ``triple support
inequality'',
\begin{equation}
{\cal S}(d;3)\equiv\norm{\psi}_{0}+\norm{F^{\dagger}\psi}_{0}+\norm{H_{2}^{\dagger}\psi}_{0}\geq\frac{3}{2}(d+1)\,.\label{eq: triple sum inequality}
\end{equation}
Clearly, this inequality cannot be saturated for dimension $d=2$
because the overall support ${\cal S}(2)$ of a state is always an
integer number. Taking only two possible values, the smallest achievable
value of the triple support size ${\cal S}(2;3)\equiv{\cal S}(2)$
equals $T_{s}(2)=5$; here and in the following, \emph{achievable}---or
\emph{sharp}---bounds of $\mathcal{S}(d)$ are denoted by $T_{s}(d)$.

The lower bound on the triple uncertainty relation for continuous
variables \cite{kechrimparis_heisenberg_2014}, derived similarly
by combining pair uncertainty relations, can also not be reached.
Theorem \ref{thm:multi-bound} is not constructive, hence it is not
obvious whether the case of $d=2$ represents an exception or whether
the inequalities (\ref{eq:extended sum UR}) are never sharp. In the
next section we will first derive some general results about multiple-support
inequalities, followed by a closer look at dimensions $3\leq d\leq19$.

\section{Saturating support inequalities for MU bases \label{sec:Saturating-the-inequality}}

To saturate the bound of the inequality (\ref{eq:extended sum UR})
means to identify states that minimise all support pair relations
simultaneously. We present a number of rigorous results for prime
dimensions $d\leq7$. Numerical methods are then used to determine
whether the generalised inequality can be saturated for dimensions
up to $d=19$.

\subsection{Constraints on saturating states \label{subsec:Results for any d}}

Our first general result is a necessary and sufficient condition that
the support inequality (\ref{eq:extended sum UR}) involving a complete
set of $(d+1)$ MU bases be saturated.
\begin{thm}[Equal support sizes]
\label{thm: Condition on 0-norms for saturation}The additive support
inequality for a complete standard set of MU bases (\ref{eq:extended sum UR})
is saturated by a state $\psi\in{\cal H}_{d}$ if and only if it has
the same support in all $(d+1)$ MU bases, i.e.
\begin{equation}
||\psi||_{j}=\frac{d+1}{2}\,,\qquad j\in\left\{ 0...\,d\right\} \,,\label{eq: saturation by equi-support states}
\end{equation}
where $d$ is an odd prime.
\end{thm}
\begin{proof}
Substituting the values (\ref{eq: saturation by equi-support states})
into (\ref{eq:extended sum UR}) directly produces the lower bound. 

For the converse, we show that the supports must have the values given
in (\ref{eq: saturation by equi-support states}) if equality holds
in Eq. (\ref{eq:extended sum UR}). Noting that the support of any
state $\psi\in{\cal H}_{d}$ ranges from $1$ to $d$, i.e.
\begin{equation}
\norm{\psi}_{0}=\frac{d+1}{2}\pm\Delta\,,\qquad\Delta\in\left\{ 0,1,\ldots,\,\frac{1}{2}(d-1)\right\} \,,\label{eq:saturation by equi-support: part 1}
\end{equation}
we will proceed by exhausting all its values in the computational
basis $\mathcal{B}_{0}$. It turns out that the the minimum in (\ref{eq:extended sum UR})
cannot be reached if the support is either $(i)$ smaller or $(ii)$
larger than $(d+1)/2$, leaving $(iii)$ the values in (\ref{eq: saturation by equi-support states})
as the only option.

$(i)$ If $\norm{\psi}_{0}=\left(d+1\right)/2-\Delta$, $\Delta>0$,
then (\ref{eq: all pair sum URs}) implies that $\norm{\psi}_{j}\geq\left(d+1\right)/2+\Delta$,
$j=\left\{ 1\ldots\,d\right\} $. Hence, the sum of the supports in
all $(d+1)$ MU bases equals 
\begin{equation}
\begin{aligned}\mathcal{S}(d)=\sum_{j=0}^{d}\norm{\psi}_{j} & \geq\frac{d+1}{2}-\Delta+d\left(\frac{d+1}{2}+\Delta\right)\\
 & \geq\frac{(d+1)^{2}}{2}+(d-1)\Delta>\frac{(d+1)^{2}}{2}\,.
\end{aligned}
\label{eq:saturation by equi-support: part 2}
\end{equation}
Therefore, the inequality cannot be saturated by a state which has
support smaller than $\left(d+1\right)/2$ in the basis $\mathcal{B}_{0}$.

$(ii)$ Assume that $\norm{\psi}_{0}=\left(d+1\right)/2+\Delta$,
$\Delta>0$. Clearly, the lower bound of the sum in Eq. (\ref{eq:extended sum UR})
can only be reached if the support of the state $\psi$ is smaller
than $(d+1)/2$ in at least one of the MU bases, $\norm{\psi}_{j^{*}}<\left(d+1\right)/2$,
$j^{*}\in\left\{ 1\ldots\,d\right\} $, say. Repeating the argument
from $(i)$ relative to the MU basis $\mathcal{B}_{j^{*}}$ instead
of $\mathcal{B}_{0}$ implies that the inequality (\ref{eq:extended sum UR})
cannot be saturated. 

$(iii)$ If $\norm{\psi}_{0}=\left(d+1\right)/2$ then (\ref{eq: all pair sum URs})
implies that $\norm{\psi}_{j}\geq\left(d+1\right)/2$, $j=\left\{ 1\ldots\,d\right\} $.
However, given these bounds, the minimum of $\mathcal{S}(d)$ in (\ref{eq:extended sum UR})
can be achieved only if the support of the state $\psi\in{\cal H}_{d}$
takes the value $\left(d+1\right)/2$ in all other MU bases as well. 
\end{proof}
The second general result states that a specific $d$-th root of unity
can appear at most \emph{twice} in the columns of the Hadamard matrices
$H_{j}$, $j=2\ldots d$, given in (\ref{eq:hadamard matrices defn}).
The proof of another necessary--but not sufficient--condition for
saturating the generalised inequality (\ref{eq:extended sum UR})
will rely on this limit of the occurrences of roots. The property
also applies to $H_{1}=F$ where each root appears exactly once in
every column, as is seen by inspecting (\ref{eq: Fourier matrix}).

\begin{restatable}[\textit{Frequency of roots}]{lem}{frequencyofroots}

\label{lem:Frequency of roots} Let $d$ be prime and consider the
states $\ket{\phi_{k}^{j}}$, $j=2,\ldots,\,d$, in Eq.\emph{ }(\ref{eq:MU basis states odd prime d})
forming the bases $\mathcal{B}_{j}$ which are MU to both the identity
and the Fourier matrix. Any $d$-th root $\omega^{n}$, $n\in\{0\ldots d\}$,
figures at most twice among the numbers $\sqrt{d}\langle x\ket{\phi_{k}^{j}}$,
$x\in\left\{ 0\ldots\:d-1\right\} $.

\end{restatable}
\begin{proof}
We need to determine the number of solutions of the equation $\omega^{-kx+\left(j-1\right)x^{2}}=\omega^{n}$
which becomes $\left(j-1\right)x^{2}-kx-n\mod d=0$ upon taking the
logarithm and rearranging. Since $j\neq1$, the equation is quadratic
for each $n$ and there can be at most two integer solutions for the
unknown $x$. The extension to the special case of $d=2$ is trivial.
\end{proof}
According to Theorem \ref{thm: Condition on 0-norms for saturation},
a state saturating (\ref{eq:extended sum UR}) must have $(d-1)/2$
vanishing expansion coefficients in each MU basis of the standard
set, in any odd prime dimension. A third general result is that there
are constraints on the distributions of these zeroes when expanded
in the MU bases of a complete set.

To spell out these constraints, let us introduce the \textit{zero
distributions} $\mathcal{Z}^{j}$ of a state $\psi\in{\cal H}_{d}$
which list the indices of the vanishing expansion coefficients in
the $\left(d+1\right)$ bases of the complete set, 
\begin{equation}
\mathcal{Z}^{j}=\left\{ \kappa\in\left\{ 0\ldots\,d-1\right\} :\braket{\phi_{\kappa}^{j}}{\psi}=0\right\} \,,\qquad j=0,\ldots,\,d\,.\label{eq:zero-distribution}
\end{equation}
Using the relation $\braket{\phi_{\kappa}^{j}}{\psi}=\bra{\phi_{\kappa}^{0}}H_{j}^{\dagger}\ket{\psi}$,
one can also think of $\mathcal{Z}^{j}$ as the set of vanishing coefficients
of the state $H_{j}^{\dagger}\ket{\psi}$ in the computational basis.

Two zero distributions of vectors in the same Hilbert space are said
to be \textit{compatible}\textit{\emph{,}}\emph{ }$\mathcal{Z}\sim\mathcal{Z}^{\prime}$,
if they are equal up to a cyclic shift. In other words, two compatible
distributions $\mathcal{Z}=\left\{ \kappa_{1},\kappa_{2},...,\kappa_{\delta}\right\} $
and $\mathcal{Z}^{\prime}=\left\{ \kappa_{1}^{\prime},\kappa_{2}^{\prime},...,\kappa_{\delta}^{\prime}\right\} $
must have the same number $\delta$ of elements and the mapping $\kappa_{i}\mapsto\kappa_{i}+\mu\,\mod d$
for some fixed integer $\mu$ must be a bijection from $\mathcal{Z}$
to $\mathcal{Z}^{\prime}$. Compatibility of zero distributions is
an equivalence relation between classes of $d$ elements.

The extension of Chebotarëv's Theorem shown in Sec. \ref{subsec:Support-inequality-MU-pairs}
and Lemma \ref{lem:Frequency of roots} imply a constraint on zero
distributions for all prime dimensions $d>3$. This property will
be used in Sec. \ref{subsec:d=00003D5,7} to prove that the support
inequality (\ref{eq:extended sum UR}) \emph{cannot} be saturated
in dimensions $d=5$ and $d=7$.
\begin{thm}
\label{thm:No compatible zero-distributions}Let $d>3$ be prime and
$\psi\in\mathcal{H}_{d}$ be a state with $(d-1)/2$ expansion coefficients
vanishing in the computational basis and in two more standard MU bases,
i.e.
\begin{equation}
\norm{\psi}_{0}=\norm{\psi}_{j_{1}}=\norm{\psi}_{j_{2}}=\frac{d+1}{2}\,,\qquad j_{1}>j_{2}\neq0\,.\label{eq: equals supports}
\end{equation}
Then the zero distributions associated with the vectors $H_{j_{1}}^{\dagger}\ket{\psi}$
and $H_{j_{2}}^{\dagger}\ket{\psi}$, respectively, are incompatible.
\end{thm}
\begin{proof}
Since the state $\psi$ has $d_{-}\equiv(d-1)/2$ vanishing components
in three bases with labels $j=0,j_{1},j_{2}$, it satisfies $3d_{-}$
conditions,
\begin{align}
\braket{\phi_{\kappa_{1}^{0}}^{0}}{\psi} & =\braket{\phi_{\kappa_{2}^{0}}^{0}}{\psi}=\ldots=0\,,\qquad\mathcal{Z}^{0}=\left\{ \kappa_{1}^{0},\kappa_{2}^{0},\ldots,\kappa_{d_{-}}^{0}\right\} \,,\nonumber \\
\braket{\phi_{\kappa_{1}^{1}}^{j_{1}}}{\psi} & =\braket{\phi_{\kappa_{2}^{1}}^{j_{1}}}{\psi}=\ldots=0\,,\qquad\mathcal{Z}^{j_{1}}=\left\{ \kappa_{1}^{1},\kappa_{2}^{1},\ldots,\kappa_{d_{-}}^{1}\right\} \,,\label{eq:scalar products - Thm 6}\\
\braket{\phi_{\kappa_{1}^{2}}^{j_{2}}}{\psi} & =\braket{\phi_{\kappa_{2}^{2}}^{j_{2}}}{\psi}=\ldots=0\,,\qquad\mathcal{Z}^{j_{2}}=\left\{ \kappa_{1}^{2},\kappa_{2}^{2},\ldots,\kappa_{d_{-}}^{2}\right\} \,.\nonumber 
\end{align}
We proceed by contradiction. To assume that the zero distributions
$\mathcal{Z}^{j_{1}}$ and $\mathcal{Z}^{j_{2}}$ \emph{are} compatible
means that they are related by a cyclic shift by some integer $\mu\in\{0\ldots d-1\}$.
In particular, we can arrange the elements in the two sets such that
\begin{equation}
\kappa_{i}^{2}=\kappa_{i}^{1}+\mu\quad\text{mod}\,d\qquad\text{for all}\quad i\in\left\{ 1\ldots d_{-}\right\} .\label{eq:compatibility condition Thrm 6}
\end{equation}
Then, according to Eq. (\ref{eq: relating columns of H_j}), the corresponding
states must be related by powers of the matrices $D$ and $B$, 
\begin{equation}
\ket{\phi_{\kappa_{i}^{2}}^{j_{2}}}=D^{j_{2}-1}B^{\kappa_{i}^{2}}\ket{\phi_{0}^{1}}=D^{j_{2}-1}B^{\kappa_{i}^{2}}D^{-j_{1}+1}B^{-\kappa_{i}^{1}}\ket{\phi_{\kappa_{i}^{1}}^{j_{1}}}=D^{j_{2}-j_{1}}B^{\mu}\ket{\phi_{\kappa_{i}^{1}}^{j_{1}}}\,,\label{eq: relate columns of MU bases-1}
\end{equation}
where we have used the fact that $D$ and $B$ commute. Defining $V_{\mu}^{\dagger}=D^{j_{2}-j_{1}}B^{\mu}$,
the third set of conditions in (\ref{eq:scalar products - Thm 6})
turns into
\begin{equation}
\braket{\phi_{\kappa_{i}^{2}}^{j_{2}}}{\psi}=\braket{\phi_{\kappa_{i}^{1}}^{j_{1}}}{V_{\mu}\psi}=0\qquad\text{for all}\quad i\in\left\{ 1\ldots d_{-}\right\} .\label{eq:scalar products set 1 for U=00005Cpsi}
\end{equation}
Since $V_{\mu}$ is diagonal in the computational basis, we have
\begin{equation}
\braket{\phi_{\kappa_{i}^{0}}^{0}}{\psi}=\braket{\phi_{\kappa_{i}^{0}}^{0}}{V_{\mu}\psi}=0\qquad\text{for all}\quad i\in\left\{ 1\ldots d_{-}\right\} \,,\label{eq:scalar products set 2 for U=00005Cpsi}
\end{equation}
which means that $\mathcal{Z}^{0}$ and $\mathcal{Z}^{j_{1}}$ are
zero distributions for the pair of vectors $\psi$ and $V_{\mu}\psi$.
In other words, these two states are both orthogonal to the same set
of $2d_{-}=\left(d-1\right)$ vectors 
\begin{equation}
\left\{ \phi_{\kappa_{1}^{0}}^{0},\ldots,\phi_{\kappa_{d_{-}}^{0}}^{0},\phi_{\kappa_{1}^{1}}^{j_{1}},\ldots,\phi_{\kappa_{d_{-}}^{1}}^{j_{1}}\right\} ,\label{eq: set of OG vectors}
\end{equation}
stemming from the computational basis $\mathcal{B}_{0}$ and the basis
$\mathcal{B}_{j_{1}}$. According to Corollary \ref{cor:Linear independence of d vectors from ANY two MUBs},
this is a set of $(d-1)$ linearly independent vectors so that only
one unique ray in $\mathcal{H}_{d}$ can exist that is orthogonal
to all of them. Therefore, the vectors $\psi$ and $V_{\mu}\psi$
must be collinear, i.e. $V_{\mu}\psi=\lambda\psi$ for some non-zero
scalar $\lambda\in\mathbb{C}$.

Since $V_{\mu}=D^{j_{1}-j_{2}}B^{-\mu}$ is diagonal in $\mathcal{B}_{0}$,
the computational basis states are eigenvectors of $V_{\mu}$. By
assumption, the state $\psi$ has $d_{+}\equiv\left(d+1\right)/2$
non-zero coefficients in this basis. Thus, the state $\psi$ will
be an eigenvector of the unitary $V_{\mu}$ only if $\lambda$ is
an eigenvalue with multiplicity of $d_{+}$ (at least). However, this
is impossible for prime dimensions $d>3$: the non-zero matrix elements
on the diagonal of $V_{\mu}$ coincide with the components of the
vector$\sqrt{d}\ket{\phi_{-\mu}^{j_{1}-j_{2}+1}}$ in the computational
basis but for $j_{1}>j_{2}\neq0\,\left(\text{mod}\,d\right)$ no more
than two of the components may coincide according to Lemma \ref{lem:Frequency of roots}.
Thus, at most two of the eigenvalues of $V_{\mu}$ can coincide. No
contradiction arises for dimension $d=3$ where $\psi$ has exactly
two non-vanishing coefficients in the computational basis.
\end{proof}

\subsection{Dimension $d=3$ \label{subsec: Dimension d=00003D3}}

To prove that the bound (\ref{eq:extended sum UR}) can be achieved
in the space $\mathcal{H}_{3}$, we exhibit the states which minimise
the support inequality.
\begin{thm}
\label{thm:Saturating states in d=00003D3}The state $\psi$ saturates
the generalised support inequality (\ref{eq:extended sum UR}) in
dimension $d=3$ if and only if it is one of the following nine (non-normalised)
qutrit states,
\begin{equation}
\begin{pmatrix}1\\
-\omega^{m}\\
0
\end{pmatrix}\,,\quad\begin{pmatrix}1\\
0\\
-\omega^{m}
\end{pmatrix}\,,\quad\begin{pmatrix}0\\
1\\
-\omega^{m}
\end{pmatrix}\,,\qquad m\in\left\{ 0,1,2\right\} \,,\label{eq: states that saturate extended UR d=00003D3}
\end{equation}
with $\omega\equiv e^{2i\pi/3}$ being a third root of unity.
\end{thm}
\begin{proof}
Theorem \ref{thm: Condition on 0-norms for saturation} implies that
a state $\psi$ saturates Eq. (\ref{eq:extended sum UR}) w.r.t. a
complete standard set of MU bases if and only if it has support two
in each of them, i.e. $||\psi||_{j}\equiv\norm{H_{j}^{\dagger}\psi}_{0}=2$,
$j=0\ldots3$. First, we assume that the third component of a candidate
state vanishes in the computational basis, i.e. $\psi=\begin{pmatrix}a, & b, & 0\end{pmatrix}^{T}$,
with non-zero complex numbers $a$ and $b$. Applying the matrices
$H_{j}^{\dagger}$, $j=1,2,3$, to it, we find four vectors, 
\begin{equation}
\left(\begin{array}{c}
a\\
b\\
0
\end{array}\right)\,,\left(\begin{array}{c}
a+b\\
a+\omega b\\
a+\omega^{2}b
\end{array}\right)\,,\left(\begin{array}{c}
a+\omega^{2}b\\
a+b\\
a+\omega b
\end{array}\right)\,,\left(\begin{array}{c}
a+\omega b\\
a+\omega^{2}b\\
a+b
\end{array}\right)\,.\label{eq:saturating states d=00003D3 - part 1}
\end{equation}
The components of the last three vectors agree, except for permutations.
Hence, support size two can occur in three different ways: one component
of each vector vanishes if 
\begin{equation}
b=-\omega^{m}a\,,\qquad m\in\left\{ 0,1,2\right\} \,,\label{eq:saturating states d=00003D3 - part 2}
\end{equation}
holds for some value of $m$. After removing an irrelevant phase,
we obtain the first three vectors given in Eq. (\ref{eq: states that saturate extended UR d=00003D3}).
Second, repeating this argument for initial vectors of the form $\psi=\begin{pmatrix}a, & 0, & b\end{pmatrix}^{T}$
and $\psi=\begin{pmatrix}0, & a, & b\end{pmatrix}^{T}$, respectively,
leads to the remaining six vectors in (\ref{eq: states that saturate extended UR d=00003D3}).

Having exhausted all three-component vectors in the computational
basis with support two, we have shown that the nine vectors in (\ref{eq: states that saturate extended UR d=00003D3})
are the only states saturating the support inequality (\ref{eq:extended sum UR})
for $d=3$.
\end{proof}

\subsection{Dimensions $d=5$ and $d=7$ \label{subsec:d=00003D5,7}}

We will show that it is impossible to reach the lower bound of the
support inequality (\ref{eq:extended sum UR}) in dimensions $d=5$
and $d=7$. The proof relies on a property of the zero distributions
of the vectors $H_{j}^{\dagger}\psi$, $j=0\ldots d$, which were
introduced in Sec. \ref{subsec:Results for any d}.
\begin{thm}
\label{thm:No saturation in d=00003D5,7}The additive support uncertainty
relation (\ref{eq:extended sum UR}) cannot be saturated in dimensions
$d=5$ and $d=7$.
\end{thm}
\begin{proof}
Let $\mathcal{Z}_{n}^{d}$ be the set of the zero distributions with
$n$ zeroes among the computational-basis coefficients of qudit states
in the Hilbert space $\mathcal{H}_{d}$. These distributions are determined
by choosing $n$ out of $d$ indices; hence, there are $|\mathcal{Z}_{n}^{d}|=\binom{d}{n}$
such sets. Recalling that \emph{compatible} sets of zero distributions
form equivalence classes, obtained from rigidly shifting a given one,
only $|\mathcal{Z}_{n}^{d}/\sim|=\binom{d}{n}/d$ \emph{inequivalent}
zero distributions exist.

According to Theorem \ref{thm: Condition on 0-norms for saturation},
a state $\ket{\psi}\in\mathcal{H}_{d}$ saturating (\ref{eq:extended sum UR})
for $d>3$, must have $n=(d-1)/2$ zeroes in each basis. In addition,
a saturating state requires the existence of at least $d$ \emph{incompatible}
zero distributions as Theorem \ref{thm:No compatible zero-distributions}
does not allow compatible zero distributions for more than two bases.
In other words, the inequality $|\mathcal{Z}_{(d-1)/2}^{d}/\sim|\geq d$
must hold. Clearly, this does not happen for $d=5$ and $d=7$ since
$|\mathcal{Z}_{2}^{5}/\sim|=2<5$ and $|\mathcal{Z}_{3}^{7}/\sim|=5<7$
, respectively. When $d\geq11,$ however, the inequality is satisfied,
with $|\mathcal{Z}_{5}^{11}/\sim|=42>11$, for example.
\end{proof}

\subsection{Numerical results for $5\protect\leq d\protect\leq19$ \label{subsec:Higher-dimensions}}

For prime numbers $d$ greater than seven, more than $d$ incompatible
zero distributions exist which removes the bottleneck we exploited
to prove Theorem \ref{thm:No saturation in d=00003D5,7}. In the absence
of an analytic handle on the problem, we will use numerical means
to check whether the bound imposed by (\ref{eq:extended sum UR})
can be reached for dimensions larger than $d=7$. 

A saturating state necessarily has $(d-1)/2$ zeroes in each MU basis.
Thus, if one picks two distinct MU bases with labels $j_{1},j_{2}\in\left\{ 0\ldots\,d\right\} $,
say, with corresponding zero distributions $\mathcal{Z}^{j_{1}}$
and $\mathcal{Z}^{j_{2}}$, the state will have vanishing scalar products
with a total of $(d-1)$ states which---in view of Corollary \ref{cor:Linear independence of d vectors from ANY two MUBs}---are
known to be linearly independent. Consequently, there is a\emph{ unique}
ray $\psi^{\perp}\in\mathcal{H}_{d}$ associated with any two zero
distributions of the type considered. If the support size of the states
$\psi^{\perp}$ generated in this way (i.e. for all possible choices
of initial zero distributions $\mathcal{Z}^{j_{1}}$ and $\mathcal{Z}^{j_{2}}$)
is always larger than $(d+1)/2$ in some third MU basis, then the
support inequality (\ref{eq:extended sum UR}) cannot be saturated:
if no state with support size $(d+1)/2$ in \emph{three }MU bases
exists, then no state with support size $(d+1)/2$ in \emph{$(d+1)$
}MU bases will exist. Since only a finite number of zero distributions
needs to be checked for a given dimension $d$, this approach actually
represents an \emph{algorithm} to check whether the lower bound can
be reached.

Running the program for prime numbers with $5\leq d\leq19$ means
to check an exponentially increasing number of cases. On a standard
PC, the program ran about a second for $d=5$ and $d=7$ while it
took about a week for $d=17$. No state has been found which would
display $(d-1)/2$ zeroes in three MU bases. For dimensions $d=5$
and $d=7$, this result is stronger than that of Sec. \ref{subsec:d=00003D5,7}
since the non-existence of a state with two and three zeroes, respectively,
is sufficient to derive Theorem \ref{thm:No saturation in d=00003D5,7},
but not \emph{vice versa}. Due to the exponential increase in the
number of zero distributions, dimensions larger than $d=19$ were
out of of reach.

\subsection{Dimensions $d>19$\label{subsec:Dimensions d>19-1}}

To satisfy the additive support inequality (\ref{eq:extended sum UR})
relative to $(d+1)$ MU bases, a state needs to satisfy more than
one pair relation (\ref{eq: all pair sum URs}) simultaneously which
seems unlikely. It is all the more surprising that for dimension $d=3$
the bound $T(3)=8$ is actually sharp, i.e. $T_{s}(3)=T(3)$. Our
results for prime dimensions up to $d=19$ suggest that this case
is exceptional.

We conjecture that \emph{the generalised uncertainty relation (\ref{eq:extended sum UR})
in prime dimensions can only be saturated when $d=3$. }Here is a
plausibility argument to support this view. Assume that a saturating
state $\psi\in\mathcal{H}_{d}$ exists for some prime dimension $d\geq3$.
According to Theorem \ref{thm: Condition on 0-norms for saturation},
the state must be orthogonal to exactly $\left(d-1\right)/2$ vectors
from each of the $\left(d+1\right)$ MU bases. Corollary \ref{cor:Linear independence of d vectors from ANY two MUBs}
implies that orthogonality with respect to just two such sets---i.e.
$\left(d-1\right)$ vectors---already determines a unique state.
Therefore, the remaining $\left(d-1\right)^{2}/2$ vectors (one set
of $(d-1)/2$ vectors is associated with each of the $(d-1)$ MU bases
not yet considered) must all lie in the same $\left(d-1\right)$-dimensional
subspace orthogonal to the state $\psi$. This is known to happen
for $d=3$ but seems hard to satisfy for larger dimensions.

\section{Sharp lower bounds\label{sec:Achievable-bounds}}

According to the results presented in Sec. \ref{sec:Saturating-the-inequality},
no states exist which would saturate the lower bounds (\ref{eq:extended sum UR})
for the support sizes in dimensions up to $d=19$, with the exception
of $d=3$. The focus of this section will be on identifying achievable
bounds.

\subsection{Dimension $d=3$\label{subsec:Dimension d=00003D3}}

Theorem \ref{thm:Saturating states in d=00003D3} in Sec. \ref{subsec: Dimension d=00003D3}
displays the states which achieve the lower bound (\ref{eq:extended sum UR})
in dimension $d=3$. In other words, the bound for the overall support
of qutrit states $\psi$ is sharp, $\mathcal{S}(3)\geq8$, where $\mathcal{S}\left(d\right)\equiv\sum_{j=0}^{d}\norm{\psi}_{j}$
for $\psi\in\mathcal{H}_{d}$.

\subsection{Dimension $d=5$ \label{subsec: A-sharp-bound d=00003D5}}

Theorem \ref{thm:No saturation in d=00003D5,7} shows that, for any
state $\psi\in\mathcal{H}_{5}$, the overall support of the states
$H_{j}^{\dagger}\psi$, $j=\left\{ 0\ldots d\right\} $, must satisfy
$\mathcal{S}(5)>18$. In this section, we will prove a sharp lower
bound, namely $\mathcal{S}(5)\geq22\equiv T_{s}(5)$.

To begin, we generalise Lemma \ref{lem:Frequency of roots} which
will be necessary for the proof of Lemma \ref{lem:Stronger argument for d=00003D5}.

\begin{restatable}{lem}{identicalentries}

\label{lem:identical entries of MU states}Let $d$ be an odd prime
and $\omega\equiv e^{\frac{2i\pi}{d}}$. Consider two states $\ket{\phi_{k_{1}}^{j_{1}}},\ket{\phi_{k_{2}}^{j_{2}}}\in\mathcal{H}_{d}$
taken from different standard MU bases, $j_{1},j_{2}\neq0$, and let
$\left\{ \ket x\right\} $ be the computational basis. Then there
can be at most two values of $x\in\left\{ 0\ldots d-1\right\} $ such
that
\begin{equation}
\braket x{\phi_{k_{1}}^{j_{1}}}=\omega^{n}\braket x{\phi_{k_{2}}^{j_{2}}}\label{eq:at most two MUB vector entries coincide}
\end{equation}
for the same value of $n\in\left\{ 0...d-1\right\} $. If two different
states are taken from the same basis, $j_{1}=j_{2}$, then the equation
has exactly one solution for each value of $n$.

\end{restatable}
\begin{proof}
See \nameref{Appendix}.
\end{proof}
Now consider a state with two vanishing expansion coefficients in
both the computational basis and a second basis of the complete set.
It turns out that such a state can have only non-zero coefficients
in the remaining four bases, resulting in a total support size of
26.

\begin{restatable}{lem}{restrictionfivedim}

\label{lem:Stronger argument for d=00003D5} If the support of a state
$\psi\in\mathcal{H}_{5}$ equals three in both the computational basis
and another standard MU basis with label $j\neq0$, i.e. $\norm{\psi}_{0}=\norm{\psi}_{j}=3$,
then its support size in each of the remaining four bases equals five,
$\norm{\psi}_{j^{\prime}}=5$, with $j^{\prime}\neq0,j$.

\end{restatable}
\begin{proof}
The proof, given in \nameref{Appendix}, uses Corollary \ref{cor:Extension of Chebo to product matrices}
and Lemma \ref{lem:identical entries of MU states}.
\end{proof}
This result allows us to determine a sharp bound $T_{s}(5)$ for the
support size $\mathcal{S}(5)$.
\begin{thm}
Given a state $\psi\in\mathcal{H}_{5}$, the sharp bound on its overall
support size $\mathcal{S}(5)$ in a complete standard set of MU bases
is given by $T_{s}(5)=22$.
\end{thm}
\begin{proof}
To construct the bound, we go through all possible values of the support
size of the state $\psi$ in the computational basis, i.e. $\norm{\psi}_{0}\in\left\{ 1\ldots\,5\right\} $. 

For $\norm{\psi}_{0}=1$, the pair inequalities (\ref{eq: all pair sum URs})
imply that the state $\psi$ must have full support in all other five
MU bases, i.e. $\norm{\psi}_{j\neq0}=5$. Hence, the overall support
of a computational basis state is given by $\mathcal{S}(5)=26$.

For $\norm{\psi}_{0}=2$, the pair inequalities (\ref{eq: all pair sum URs})
imply that the state $\psi$ can have at most one zero in each of
the other five MU bases, i.e. $\norm{\psi}_{j\neq0}=4$. Hence, the
overall support of $\psi$ is given by $\mathcal{S}(5)=22$. All 300
states of the form 
\begin{equation}
\ket{\psi}=\frac{1}{\sqrt{2}}\left(\ket{\phi_{k_{1}}^{j}}-\omega^{n}\ket{\phi_{k_{2}}^{j}}\right),\qquad j\in\left\{ 0\ldots5\right\} ,\quad k_{1},k_{2},n\in\left\{ 0\ldots4\right\} ,\quad k_{1}\neq k_{2}\,.\label{eq:saturating states of improved d=00003D5 UR}
\end{equation}
achieve this bound. More generally, for primes $d>3$, there are 
\begin{equation}
\left(d+1\right)d\binom{2}{d}=\frac{1}{2}\left(d^{2}-1\right)d^{2}\label{eq: number of states}
\end{equation}
such states as $j$ takes $(d+1)$ values, $n$ takes $d$ values
and there are $\binom{2}{d}$ different pairs of $k_{1}$ and $k_{2}$.
The three-dimensional case is an exception, as demonstrated by Theorem
\ref{thm:Saturating states in d=00003D3}.

For $\norm{\psi}_{0}=3$, the pair inequalities (\ref{eq: all pair sum URs})
rule out a support size lower than three in any basis from the set.
We apply Lemma \ref{lem:Stronger argument for d=00003D5}: the support
size of $\psi$ can equal three in only one other MU basis while the
state must have full support in the others, leading to $\mathcal{S}(5)=26$.
It is also possible to have $\norm{\psi}_{j}=4$ in all bases but
the first one, i.e. for $j\neq0$. In this case seven expansion coefficients
would vanish over the complete set, resulting in an overall support
size of $\mathcal{S}(5)=23$. This bound is larger than the one already
obtained for the case of $\norm{\psi}_{0}=2$.

If $\norm{\psi}_{0}=4$ and all other support sizes are also equal
to four, the resulting overall support of $\mathcal{S}(5)=24$ is
again larger that the previous bound of $\mathcal{S}(5)=22$ obtained
for $\norm{\psi}_{0}=2$. To improve on the value of $\mathcal{S}(5)=24$,
at least one of the other norms must fall below four, i.e. $1\leq||\psi||_{j^{*}}\leq3$
for some $j^{*}\neq0$. This assumption, however, sends us back to
one of the cases already discussed: we formally map $j^{*}\mapsto0$
and repeat the arguments given for $1\leq||\psi||_{0}\leq3$.

Similarly, full support in all six MU bases cannot beat any of the
bounds given so far. Improving on the value of $\mathcal{S}(5)=30$
is only possible by decreasing some of the support sizes, so that
we will end up in one of the previously discussed cases. Having considered
all support sizes of a state in a basis, we have exhausted all possibilities
and conclude that the bound on the overall support of a state $\psi\in\mathcal{H}_{5}$
in six MU bases is indeed given by $T_{s}(5)=22$.
\end{proof}

\subsection{Dimension $d=7$ \label{subsec: A-sharp-bound d=00003D7}}

Our aim is to identify states which minimise the overall support $\mathcal{S}(7)=\sum_{j=0}^{7}\norm{\psi}_{j}$.
To determine the sharp bound for $d=7$, we will proceed as in the
previous section. However, since no equivalent to Lemma \ref{lem:Stronger argument for d=00003D5}
is known, we will partly rely on numerical results. 

For $\norm{\psi}_{0}=1$, the pair inequalities (\ref{eq: all pair sum URs})
imply that the state $\psi$ must have full support in all other seven
MU bases,$\norm{\psi}_{j\neq0}=7$. Hence, the overall support of
$\psi$ is given by $\mathcal{S}(7)=50$. 

For $\norm{\psi}_{0}=2$, the pair inequalities (\ref{eq: all pair sum URs})
imply that the state $\psi$ can have at most one zero in each of
the other seven MU bases, i.e. $\norm{\psi}_{j\neq0}=6$. Hence, the
overall support of $\psi$ is given by $\mathcal{S}(7)=44$, achieved
by states of the form 
\begin{equation}
\ket{\psi}=\frac{1}{\sqrt{2}}\left(\ket{\phi_{k_{1}}^{j}}-\omega^{n}\ket{\phi_{k_{2}}^{j}}\right),\qquad j\in\left\{ 0\ldots7\right\} ,\quad k_{1},k_{2},n\in\left\{ 0\ldots6\right\} ,\quad k_{1}\neq k_{2}\,.\label{eq:saturating states of improved d=00003D7 UR}
\end{equation}
According to Eq. (\ref{eq: number of states}), there are 1176 such
states. 

For $\norm{\psi}_{0}=3$, the smallest possible value of $\mathcal{S}(7)$
compatible with the pair inequalities is $\mathcal{S}(7)=38$, as
the states in the other bases must have support size at least five
each, i.e. $\norm{\psi}_{j\neq0}=5$. However, no state achieving
this bound has been found (numerically). The computations show that
a state with support sizes three and five in two MU bases must have
full support in the remaining six MU bases so that $\mathcal{S}(7)=50$.
Assuming support size six in all but the first MU basis, the overall
support would be $\mathcal{S}(7)=45$ which is higher than the bound
of $\mathcal{S}(7)=44$ achievable for $\norm{\psi}_{0}=2$ .

Given a support size of four in the first MU basis, $\norm{\psi}_{0}=4$,
not all other support sizes can be equal to four according to Theorem
\ref{thm:No saturation in d=00003D5,7}. One case corresponds a state
having support size four in the first and one other MU basis. It is
possible to (numerically) construct states for which the remaining
six supports sizes must be equal to six, leading to $\mathcal{S}(7)=44$.
We neither know analytic expressions for these states nor their total
number. The other scenario compatible with $\norm{\psi}_{0}=4$ corresponds
to the remaining seven support sizes each equalling five, i.e.$\norm{\psi}_{j\neq0}=5$,
leading to $\mathcal{S}(7)=39$ but we have obtained no evidence for
a state achieving this value.

Assume now that $\norm{\psi}_{0}=5$ and that the support of $\psi$
in the other MU bases is also at least five (we exclude all cases
with $||\psi||_{j^{*}}<5$ for some $j^{*}\neq0$ since---upon relabeling
the MU bases---they have effectively already been considered). An
overall support of $\mathcal{S}(7)=40$ results, below the previously
obtained value of $\mathcal{S}(7)=44$ for $\norm{\psi}_{0}=2$. Numerically
searching for states achieving this bound, we find that pairs of states
with support size five in two MU bases exist but no triples, ruling
out the value $\mathcal{S}(7)=40$. Assuming support size five in
two bases and at least six in the remaining six MU bases leads to
a higher support size, $\mathcal{S}(7)=46$.

Starting out with a support size of $\norm{\psi}_{0}>5$, no smaller
lower bound will exist if all other support sizes take a value of
at least six as $\mathcal{S}(7)\geq48$ follows immediately. If not
all support sizes take a value of at least six we are being sent back
to a previously discussed case. Thus, we have established the sharp
bound of $T_{s}(7)=44$ on the overall support of seven-component
vectors in eight MU bases, partly relying on numerics.

\section{Summary and Conclusions \label{sec:Summary-and-Conclusions}}

Tao's uncertainty relation provides a lower bound on the sum of the
support sizes of a state $\psi\in\mathcal{H}_{d}$ in the standard
basis and its Fourier transform, for prime dimensions $d$. By generalising
the bound to arbitrary pairs of mutually unbiased bases (cf. Theorem
\ref{thm: pair sums}), we show in Theorem \ref{thm:multi-bound}
that the sum of the support sizes of a state $\psi$ in a complete
standard set of $\left(d+1\right)$ MU bases cannot fall below $T(d)\equiv(d+1)^{2}/2$.
The bound is found to be sharp for $d=3$, and proofs were given that
it cannot be saturated for dimensions $d=2$, $5$ and $7$. Numerical
results indicate that no states exist which achieve the bound for
prime numbers up to $d\leq19$. Table \ref{tab: summary S_d} summarises
these results. We conjecture that the inequality is saturated in dimension
$d=3$ only.

\begin{table}[H]
\begin{centering}
\begin{tabular}{|c|c|c|c|c|c|c|c|c|}
\hline 
$d$ & 2 & 3 & 5 & 7 & 11 & 13 & 17 & 19\tabularnewline
\hline 
\hline 
$T(d)$ & 9/2 & 8 & 18 & 32 & 72 & 98 & 162 & 200\tabularnewline
\hline 
$T(d)$ achievable? & $\times$ & $\checkmark$ & $\times$ & $\times$ & $(\times)$ & $(\times)$ & $(\times)$ & $(\times)$\tabularnewline
\hline 
$T_{s}(d)$ & 5 & 8 & 22 & (44) & ? & ? & ? & ?\tabularnewline
\hline 
\end{tabular}
\par\end{centering}
\caption{\label{tab: summary S_d}Lower and sharp bounds $T(d)$ and $T_{s}(d)$,
respectively, on the support sizes of states $\psi\in\mathcal{H}_{d}$
when expanded in the complete standard set of $(d+1)$ MU bases, for
small prime dimensions (numerical results in parentheses).}
\end{table}

Tao's pair support inequality has been used to identify \emph{KD-nonclassical}
states, i.e. states whose Kirkwood-Dirac quasiprobability distribution
has negative or complex contributions \cite{de_bievre_complete_2021}.
Given two orthonormal bases of a finite-dimensional space $\mathcal{H}_{d},d\in\mathbb{N},$
with no common elements, a state $\psi$ is found to be KD-nonclassical
if the sum of its support sizes in these bases is greater than $\left(d+1\right)$.
KD-classicality is readily generalised to complete sets of MU bases
instead of pairs only. In this context, the results of Sec. \ref{sec:Saturating-the-inequality}
mean that no states exist which are KD-classical with respect to the
standard set of $(d+1)$ MU bases in small prime dimensions. When
$d=3$, the claim follows by directly computing the complex KD distributions
of the nine minimal uncertainty states of Eq. (\ref{eq: states that saturate extended UR d=00003D3}).

The uncertainty of quantum states involving more than two MU bases
has been studied before. Building on a result for a pair of mutually
unbiased observables \cite{maassen_generalized_1988}, entropic uncertainty
relations have been found which involve $(d+1)$ MU bases \cite{ivanovic_inequality_1992,sanchez_entropic_1993}.
Similarly, Heisenberg's uncertainty relation for continuous variables
has a counterpart based on \emph{three} observables satisfying the
canonical commutation relation pairwise \cite{kechrimparis_heisenberg_2014}.
Often, the generalisations are straightforward but the resulting inequalities
tend not to be achievable. Sharp bounds and the minimising states
are usually difficult to find (see e.g. \cite{wu_entropic_2009,kechrimparis_geometry_2017}
and the review \cite{wehner_entropic_2010}). In this respect, the
additive inequality proposed here is no exception.

Support uncertainty relations for multiple MU bases have many interesting
features. As for the pair inequalities, a \emph{finite} number of
measurements can be sufficient to confirm that a quantum state satisfies
a specific bound. The minimum number of required measurements is simply
given by the value of the relevant bound, be it sharp or not: it is
sufficient that $T(d)$ different outcomes be registered when measurements
in the MU bases are performed on the state $\psi$. This property
also ensures that KD-nonclassicality may sometimes be detected with
a finite number of measurements.

Furthermore, the lower bounds of support inequalities neither depend
on the state considered nor on the value of Planck's constant. The
absence of $\hbar$ as a parameter suggests that no support inequalities
for continuous variables will emerge in the limit of systems with
ever larger dimensions $d$. The maximal support size of a quantum
state grows without bound and, therefore, does not approach a well-defined
quantitative measure for uncertainty. Finally, we would like to point
out that determining bounds on support sizes is technically difficult
since they are basis-dependent quantities.

Establishing sharp bounds for dimensions $d\geq11$ remains an open
question which will require new insights since numerical approaches
become unfeasible with increasing dimensions. Other directions of
future work will be to study support uncertainty relations for smaller
sets of MU bases such as triples, for example. The simplification
stems from the considerably smaller number of parameters in comparison
to complete MU sets. Preliminary analytical and numerical results
for small prime dimensions $3\leq d\leq19$ suggest that no state
can saturate the bound $T(d;3)$ on the triple uncertainty relation
(\ref{eq: triple sum inequality}) for $d\neq3$.
\begin{acknowledgement*}
VF would like to thank the WW Smith fund at the University of York
for financial support. The existence of an alternative proof of Lemma
\ref{lem: reduction of product of Hadamards} mentioned in the Appendix
has kindly been pointed out by a referee.
\end{acknowledgement*}

\section{Appendix\label{Appendix}}

We present proofs of Lemma \ref{lem: reduction of product of Hadamards},
Corollary \ref{cor:Linear independence of d vectors from ANY two MUBs}
and Lemmata \ref{lem:identical entries of MU states} and \ref{lem:Stronger argument for d=00003D5},
in this order.

\productHadamardlemma*
\begin{proof}
Using Eqs. (\ref{eq:hadamard matrices defn}), we calculate the matrix
elements of the product $H_{k}^{\dagger}H_{j}$, with $j,k\neq0$
and $j\neq k$, 
\begin{equation}
\left[H_{k}^{\dagger}H_{j}\right]_{\ell\ell^{\prime}}=\braket{\phi_{\ell}^{k}}{\phi_{\ell^{\prime}}^{j}}=\frac{1}{\sqrt{d}}G_{d}(j-k,\ell-\ell^{\prime})\,,\label{eq:els of HkHj}
\end{equation}
with the generalised Gauss sum \cite{berndt_gauss_1998} 
\begin{equation}
G_{d}(j,\ell)=\frac{1}{\sqrt{d}}\sum_{x=0}^{d-1}\omega^{jx^{2}+\ell x}\,.\label{eq: generalized Gauss sum identity}
\end{equation}
Using $1=\omega^{\left(\ell-\ell^{\prime}\right)^{2}\chi}\omega^{-\left(\ell-\ell^{\prime}\right)^{2}\chi}$
in (\ref{eq:els of HkHj}) and letting $\chi$ be an integer satisfying
$4\left(j-k\right)\chi\equiv1\,\mod d$, we obtain a standard Gauss
sum $G_{d}\left(j-k,0\right)$ with known closed form. Explicitly,
for $j\neq k$, we obtain
\begin{align}
G_{d}(j-k,\ell-\ell^{\prime}) & =\omega^{-\left(\ell-\ell^{\prime}\right)^{2}\chi}\frac{1}{\sqrt{d}}\sum_{x}\omega^{\left(j-k\right)x^{2}+\left(\ell-\ell^{\prime}\right)x}\omega^{\left(\ell-\ell^{\prime}\right)^{2}\chi}\nonumber \\
 & =\omega^{-\left(\ell-\ell^{\prime}\right)^{2}\chi}\frac{1}{\sqrt{d}}\sum_{x}\omega^{\left(j-k\right)\left[x+\overline{2}\overline{\left(j-k\right)}\left(\ell-\ell^{\prime}\right)\right]^{2}}\nonumber \\
 & =\omega^{-\left(\ell-\ell^{\prime}\right)^{2}\chi}\left[\frac{1}{\sqrt{d}}\sum_{x}\omega^{\left(j-k\right)x^{2}}\right]=\omega^{-\left(\ell-\ell^{\prime}\right)^{2}\chi}G_{d}\left(j-k,0\right)\nonumber \\
 & =\omega^{-\left(\ell-\ell^{\prime}\right)^{2}\chi}\left(\frac{j-k}{d}\right)\varepsilon_{d}\,,\label{eq:Generalised Gauss sum solution}
\end{align}
where $\bar{a}$ denotes the multiplicative inverse of $a$, $\bar{a}a\equiv1$
mod $d$, while $\left(\frac{a}{b}\right)$ denotes the Jacobi symbol
of the integers $a$ and $b$, and 
\begin{equation}
\varepsilon_{d}=\begin{cases}
1 & \text{if}\;d\equiv1\,\mod4\,,\\
i & \text{if}\;d\equiv3\,\mod4\,.
\end{cases}\label{eq:epsilon factor gauss sum-1}
\end{equation}
The sum $G_{d}\left(j-k,\ell-\ell^{\prime}\right)$ in (\ref{eq:Generalised Gauss sum solution})
reduces to a phase factor as it should since the components of the
matrix $H_{k}^{\dagger}H_{j}$ are given by the overlap of states
stemming from different MU bases.

Combining (\ref{eq:els of HkHj}) and (\ref{eq:hadamard matrices defn}),
we now determine the elements of the matrix $V\equiv H_{k}^{\dagger}H_{j}H_{t}$
for arbitrary $t\neq0$:
\begin{equation}
V_{\ell\ell^{\prime}}=\sum_{\ell^{\prime\prime}=0}^{d-1}\braket{\phi_{\ell}^{k}}{\phi_{\ell^{\prime\prime}}^{j}}\braket{\ell^{\prime\prime}}{\phi_{\ell^{\prime}}^{t}}=\frac{1}{d}\sum_{\ell^{\prime\prime}=0}^{d-1}G_{d}(j-k,\ell-\ell^{\prime\prime})\,\omega^{-\ell^{\prime}\ell^{\prime\prime}+\left(t-1\right)\ell^{\prime\prime}{}^{2}}\,.\label{eq:proof new lemma - part 1}
\end{equation}
We can simplify this expression by substituting (\ref{eq:Generalised Gauss sum solution})
into it, to find 
\begin{equation}
\begin{aligned}V_{\ell\ell^{\prime}} & =\frac{1}{d}\left(\frac{j-k}{d}\right)\varepsilon_{d}\,\sum_{\ell^{\prime\prime}=0}^{d-1}\omega^{-\left(\ell-\ell^{\prime\prime}\right)^{2}\chi}\,\omega^{-\ell^{\prime}\ell^{\prime\prime}+\left(t-1\right)\ell^{\prime\prime}{}^{2}}\\
 & =\frac{1}{d}\left(\frac{j-k}{d}\right)\varepsilon_{d}\,\omega^{-\ell^{2}\chi}\sum_{\ell^{\prime\prime}=0}^{d-1}\omega^{\left(t-1-\chi\right)\ell^{\prime\prime}{}^{2}+\left(2\ell\chi-\ell^{\prime}\right)\ell^{\prime\prime}}\,.
\end{aligned}
\label{eq:proof new lemma - part 2}
\end{equation}
Letting $t=1+\chi$, we obtain sums over all $d$-th roots of one
which vanish unless the exponents of $\omega$ vanish, 
\begin{equation}
\sum_{\ell^{\prime\prime}=0}^{d-1}\omega^{\left(2\ell\chi-\ell^{\prime}\right)\ell^{\prime\prime}}=\begin{cases}
d & \text{if}\quad\ell^{\prime}=2\ell\chi\mod d\,,\\
0 & \text{otherwise}\,.
\end{cases}\label{eq:proof new lemma - part 3}
\end{equation}
Thus, for this value of $t$, the matrix elements of $V$ take the
form
\begin{equation}
V_{\ell\ell^{\prime}}=\begin{cases}
\left(\frac{j-k}{d}\right)\varepsilon_{d}\,\omega^{-\ell^{2}\chi} & \text{if}\quad\ell^{\prime}=2\ell\chi\mod d\,,\\
0 & \text{otherwise}\,,
\end{cases}\label{eq:proof new lemma - part 4}
\end{equation}
so that the only non-zero elements of the matrix $V$ are those with
indices $\left(\ell,2\ell\chi\mod d\right)$. Each row $\ell$ has
exactly one non-zero entry and the map $\ell\mapsto2\ell\chi\mod d$
constitutes a permutation of the elements of $\left\{ 0\ldots\,d-1\right\} $
since $d$ is a prime number and $2\chi\neq0\,\mod d$. (Assume this
was not the case, i.e. $2\chi x\,\mod d=2\chi y\,\mod d$ for some
$x,y\in\left\{ 0\ldots\,d-1\right\} $, $x\neq y$. Then $2\chi\left(x-y\right)=nd$
for some integer $n$ which is never the case whenever $d$ is prime
and $\chi\neq0$ mod $d$.) As a consequence, each column will also
display exactly one non-zero entry.

Therefore, the product $V$ of three Hadamard matrices is equal to
a monomial matrix $M(j,k)$ if $t=1+\chi$, i.e.
\begin{equation}
H_{k}^{\dagger}H_{j}=M(j,k)H_{1+\chi}^{\dagger}\,.\label{proof new lemma - part 5}
\end{equation}

We complete the proof by showing that the matrix $V=H_{k}^{\dagger}H_{j}H_{t}$
is \emph{not} monomial for any other value of $t$. For $t\neq1+\chi$,
the sum on the right-hand side of (\ref{eq:proof new lemma - part 2})
represents another generalised Gauss sum so that
\begin{equation}
V_{\ell\ell^{\prime}}=\frac{1}{\sqrt{d}}\left(\frac{j-k}{d}\right)\varepsilon_{d}\,\omega^{-\ell^{2}\chi}G_{d}\left(t-1-\chi,2\ell\chi-\ell^{\prime}\right)\label{eq:proof new lemma - part 6}
\end{equation}
with $G_{d}\left(t-1-\chi,2\ell\chi-\ell^{\prime}\right)=\sqrt{d}\braket{\phi_{2\ell\chi}^{1+\chi}}{\phi_{\ell^{\prime}}^{t}}$.
We can now substitute the expression (\ref{eq:Generalised Gauss sum solution})
and obtain
\begin{equation}
V_{\ell\ell^{\prime}}=\frac{1}{\sqrt{d}}\varepsilon_{d}^{2}\left(\frac{j-k}{d}\right)\left(\frac{t-1-\chi}{d}\right)\omega^{-\ell^{2}\chi}\omega^{-(2\ell\chi-\ell^{\prime})^{2}\widetilde{\chi}}\label{eq:proof new lemma - part 7}
\end{equation}
where $\widetilde{\chi}\in\left\{ 1\ldots\,d\right\} $ is an integer
satisfying $4\left(t-1-\chi\right)\widetilde{\chi}=1\mod d$. Hence,
the matrix elements $V_{\ell\ell^{\prime}}$ are all non-zero confirming
that the matrix $V$ is not monomial unless $t=1+\chi$.
\end{proof}
An alternative, shorter proof of Lemma \ref{lem: reduction of product of Hadamards}
can be given by representing the Hadamard matrices $H_{j}$ as $2\times2$
matrices in \textbf{$\text{\textbf{SL}}(2,\mathbb{Z}/d\mathbb{Z})$
}(cf. \cite{neuhauser_explicit_2002}).

\linearindepcorollary*
\begin{proof}
Construct a matrix $M$ of order $d\times\left(d_{1}+d_{2}\right)$
from any $\left(d_{1}+d_{2}\right)\leq d$ column vectors---expressed
in the computational basis---from the two MU bases $\mathcal{B}_{j_{1}}$
and $\mathcal{B}_{j_{2}}$. Then left-multiply $M$ by $H_{j_{1}}^{\dagger}$.
Since $\bra xH_{j_{1}}^{\dagger}\ket{\phi_{k}^{j_{1}}}=\braket xk$,
the first $d_{1}$ columns will be elements of the computational basis,
while the remaining $d_{2}$ columns will be taken from $H_{j_{1}}^{\dagger}H_{j_{2}}$
since $\bra xH_{j_{1}}^{\dagger}\ket{\phi_{k}^{j_{2}}}=\bra xH_{j_{1}}^{\dagger}H_{j_{2}}\ket k$.
By swapping rows appropriately via a permutation operator $P$ which
does not change linear independence of column vectors, the top left
square can be mapped to the $d_{1}$-dimensional identity. For example,
if we consider $d=5$ and $d_{1}=d_{2}=2$, we obtain a $5\times4$
matrix, 
\begin{equation}
PH_{j_{1}}^{\dagger}M=\left(\begin{array}{cc|cc}
1 & 0 & \ast & \ast\\
0 & 1 & * & *\\
\hline 0 & 0 & \ast & \ast\\
0 & 0 & \ast & \ast\\
0 & 0 & * & *
\end{array}\right)\,,\label{eq:PHM matrix}
\end{equation}
where the asterisks refer to the elements of $H_{j_{1}}^{\dagger}H_{j_{2}}$.

The $\left(d_{1}+d_{2}\right)\leq d$ vectors are linearly dependent
only if $M$ does \textit{not} have full rank, i.e. $\text{rank}\left(M\right)<d_{1}+d_{2}$.
Since $PH_{j_{1}}^{\dagger}$ is unitary, it follows that $\text{rank}\left(PH_{j_{1}}^{\dagger}M\right)<d_{1}+d_{2}$.
Given the form of the matrix (\ref{eq:PHM matrix}), the bottom-right
part of $PH_{j_{1}}^{\dagger}M$ must contain a $\left(d_{2}\times d_{2}\right)$
submatrix with vanishing determinant. However, this is prohibited
by Corollary \ref{cor:Extension of Chebo to product matrices} which
ensures for all prime numbers $d$ that $H_{j_{1}}^{\dagger}H_{j_{2}}$
has non-vanishing minors if $j_{1}\neq j_{2}$. Thus, all $(d_{1}+d_{2})$
column vectors of $M$ must be linearly independent.
\end{proof}
\identicalentries*
\begin{proof}
By substituting (\ref{eq:MU basis states odd prime d}) into (\ref{eq:at most two MUB vector entries coincide})
and taking the logarithm, one obtains $ax^{2}+bx-n=0\:\text{mod}\,d$
where $a=j_{1}-j_{2}$ and $b=k_{2}-k_{1}$. This quadratic equation
can have no more than two integer solutions. Thus, at most two components
of the states $\ket{\phi_{k_{1}}^{j_{1}}}$ and $\ket{\phi_{k_{2}}^{j_{2}}}$
can be identical in the computational basis, up to multiplication
by $\omega^{n}$. If $j_{1}=j_{2}$, then $a=0$ and the equation
is linear with a single solution for each value of $n$.
\end{proof}
\restrictionfivedim*
\begin{proof}
Four scalar products with the state $\psi$ vanish,
\begin{equation}
\bk{x_{1}}{\psi}=\bk{y_{1}}{\psi}=\bk{\phi_{x_{2}}^{j}}{\psi}=\bk{\phi_{y_{2}}^{j}}{\psi}=0\,,\label{eq: four zeros}
\end{equation}
two for each of the bases. Hence, the zero distributions of the states
$\psi$ and $H_{j}^{\dagger}\psi$, are given by $\mathcal{Z}^{0}=\left\{ x_{1},y_{1}\right\} $
and $\mathcal{Z}^{j}=\left\{ x_{2},y_{2}\right\} $ respectively,
with four integer numbers $x_{1},\ldots,y_{2}\in\{0\ldots4\}$. Now
suppose that there is a third MU basis $\mathcal{B}_{j^{\prime}}$,
different from both $\mathcal{B}_{0}$ and $\mathcal{B}_{j}$, in
which the state $\psi$ does \emph{not} have full support. In other
words, there is at least one vanishing scalar product, $\bk{\phi_{x_{3}}^{j^{\prime}}}{\psi}=0$,
say, where $x_{3}\in\{0\ldots4\}$. Expressing the components of the
five vectors $\ket{x_{1}}$, $\ket{y_{1}}$, $\ket{\phi_{x_{2}}^{j}}$,
$\ket{\phi_{y_{2}}^{j}}$, $\ket{\phi_{x_{3}}^{j^{\prime}}}$ with
respect to the computational basis and arranging them into a $5\times5$
matrix, we find, after permuting the rows and rephasing the last three
vectors,
\begin{equation}
M=\frac{1}{\sqrt{5}}\left(\begin{array}{ccccc}
\sqrt{5} & 0 & * & * & *\\
0 & \sqrt{5} & * & * & *\\
0 & 0 & 1 & 1 & 1\\
0 & 0 & \omega^{a} & \omega^{c} & \omega^{e}\\
0 & 0 & \omega^{b} & \omega^{d} & \omega^{f}
\end{array}\right)\,,\qquad a,\ldots,f\in\{0\ldots4\}\,.\label{eq: 5x5 matrix M}
\end{equation}
Being elements of the Hadamard matrices $H_{j}$ and $H_{j^{\prime}}$,
the entries of the last three columns are powers of $\omega$, a fifth
root of one. Corollary \ref{cor:Linear independence of d vectors from ANY two MUBs}
ensures the linear independence of the first four vectors.

If the determinant of $M$ does not vanish, $\det M\neq0,$ then the
five column vectors forming it are linearly \emph{independent}, thus
spanning $\mathcal{H}_{5}$. However, the only state being orthogonal
to all of $\mathcal{H}_{5}$ is $\psi=0$ which does not represent
a quantum state. Thus, for an acceptable state $\psi$ producing the
given five vanishing expansion coefficients, the five vectors involved
must be linearly \emph{dependent}, i.e. $\det M=0$. Consequently,
the determinant of the bottom right $3\times3$ matrix of $M$ must
vanish, 
\begin{equation}
\Delta\equiv\det\left(\begin{array}{ccc}
1 & 1 & 1\\
\omega^{a} & \omega^{c} & \omega^{e}\\
\omega^{b} & \omega^{d} & \omega^{f}
\end{array}\right)=\omega^{a+d}+\omega^{e+b}+\omega^{c+f}-\omega^{c+b}-\omega^{e+d}-\omega^{a+f}=0\,.\label{eq:determinant of 3x3 matrix d=00003D5 proof}
\end{equation}

Each of the six terms in this expression is a power of a fifth root
$\omega$ of unity, hence non-zero. It is well known that the set
$\left\{ \omega^{n}\,|\,n=0,...,3\right\} $ is linearly independent
over the rational numbers $\mathbb{Q}$. As a consequence, every non-zero
complex number that is expressible as a linear combination (over $\mathbb{Q}$)
of these roots of unity has a unique expression. Since $\omega^{4}=-1-\omega-\omega^{2}-\omega^{3}$,
it must follow that the only decomposition of zero over $\mathbb{Q}$
in terms of fifth roots of $1$ is $0=q\left(1+\omega+\omega^{2}+\omega^{3}+\omega^{4}\right)$,
with some rational number $q\in\mathbb{Q}$. In other words, for the
sum to vanish all five roots must be multiplied to the same rational
coefficient.

We distinguish two cases: either $q\neq0$ or $q=0$. Since (\ref{eq:determinant of 3x3 matrix d=00003D5 proof})
involves six terms with coefficients $\pm1$, we conclude that the
case of $q\neq0$ cannot be realised: it is impossible to get all
five roots to appear with the same non-zero coefficient. For example,
let $(a+d)=(e+b)\mod5$, then Eq. (\ref{eq:determinant of 3x3 matrix d=00003D5 proof})
reduces to 
\begin{equation}
\Delta=2\omega^{a+d}+\omega^{c+f}-\omega^{c+b}-\omega^{e+d}-\omega^{a+f}\label{eq:Case 1: vanishing determinant}
\end{equation}
Since all roots must appear, the exponents in (\ref{eq:Case 1: vanishing determinant})
are all different. However, the coefficients are not equal throughout
and the sum cannot vanish. A similar argument holds for any other
equality between exponents.

The case of $q=0$ must therefore apply: the determinant $\Delta$
vanishes if and only if the six terms in Eq. (\ref{eq:determinant of 3x3 matrix d=00003D5 proof})
cancel each other in pairs, i.e. the the powers of $\omega$ must
occur an even number of times, and with an equal number of positive
and negative coefficients. Hence, the first term in (\ref{eq:determinant of 3x3 matrix d=00003D5 proof})
is necessarily paired up with one of the powers with a negative coefficient
leading. Three cases arise which we will consider separately.

$(i)$ For the first and the fourth term to cancel, we must have $(a+d)=(c+b)\mod5\,,$
or 
\begin{equation}
(a-b)=(c-d)\mod5\,,\label{eq: a-b=00003Dc-d}
\end{equation}
relating the expansion coefficients of two vectors of the \emph{same}
basis, namely $\ket{\phi_{x_{2}}^{j}}$ and $\ket{\phi_{y_{2}}^{j}}$.
Consequently, the third and fourth column vectors in the matrix $M$
in (\ref{eq: 5x5 matrix M}) have (at least) two equal entries in
identical positions, up to an irrelevant common phase factor. This
would result in a vanishing $2\times2$ submatrix of $H_{j}$ contradicting
Corollary \ref{cor:Extension of Chebo to product matrices} (and Lemma
\ref{lem:identical entries of MU states}). Thus the determinant $\Delta$
cannot vanish in this case.

$(ii)$ For the first and the fifth term to cancel, we must have $(a+d)=(e+d)\mod5\,,$
or 
\begin{equation}
a=e\,\mod5\,,\label{eq: a=00003De}
\end{equation}
relating the expansion coefficients of two vectors of \emph{different}
bases, namely $\ket{\phi_{x_{2}}^{j}}$ and $\ket{\phi_{x_{3}}^{j^{\prime}}}$.
Corollary \ref{cor:Extension of Chebo to product matrices} does not
apply to this case. We do know, however, that the fourth term in the
sum (\ref{eq:determinant of 3x3 matrix d=00003D5 proof}) must pair
up with either the second or the third term of the sum in (\ref{eq:determinant of 3x3 matrix d=00003D5 proof}).
In the first case, we find $(c+b)=(e+b)\mod5\,,$or 
\begin{equation}
c=e\,\mod5\,.\label{eq: c=00003De}
\end{equation}
Given the constraint (\ref{eq: a=00003De}), we we obtain the identity
\begin{equation}
a=c\,\mod5\,,\label{eq: a=00003Dc}
\end{equation}
again relating the expansion coefficients of two vectors of the \emph{same}
basis, namely $\ket{\phi_{x_{2}}^{j}}$ and $\ket{\phi_{y_{2}}^{j}}$.
As in the Case $(i)$, a contradiction to Corollary \ref{cor:Extension of Chebo to product matrices}
arises.

In the second case, we pair up terms three and four of the sum (\ref{eq:determinant of 3x3 matrix d=00003D5 proof}),
leading to the identity $(c+b)=(c+f)\mod5$, or
\begin{equation}
b=f\,\mod5\,.\label{eq: b=00003Df}
\end{equation}
Together with Eq. (\ref{eq: a=00003De}), it follows that the last
three elements of the third and and fifth columns of $M$ are identical.
However, according to Lemma \ref{lem:identical entries of MU states},
two vectors stemming from two different bases MU to the computational
basis can have at most two identical components.

$(iii)$: Assuming that the first and the sixth term of the sum (\ref{eq:determinant of 3x3 matrix d=00003D5 proof})
cancel again leads to a contradiction along the lines of the argument
considered in Case $(ii).$

Thus, we are forced to conclude that the determinant $\Delta$ cannot
not vanish for any $j^{\prime}\neq0,j$ and any $x_{3}$, which implies
that the state $\psi$ must have full support.
\end{proof}

\end{document}